\documentclass[reqno]{amsart}
\usepackage{amsfonts}
\usepackage{mathrsfs}
\usepackage{color}
\usepackage{cite}
\usepackage{amscd}
\usepackage{latexsym}
\usepackage{amsfonts}
\usepackage{graphicx}
\usepackage{CJK,indentfirst,amsmath, amsfonts,amssymb,amsthm,cite,cases, subeqnarray,setspace,multirow}

\begin{document}
\begin{CJK*}{GBK}{song}

\title[Global existence of solutions for quantum hydrodynamics]
{Global existence and semiclassical limit for quantum hydrodynamic equations with viscosity and heat conduction}
\author[X. Pu \& B. Guo]{Xueke Pu and Boling Guo}  

\address{Xueke Pu \newline
Department of Mathematics, Chongqing University, Chongqing 401331, P.R.China} \email{ xuekepu@cqu.edu.cn}

\address{Boling Guo \newline
Institute of Applied Physics and Computational Mathematics, P.O. Box 8009,  Beijing,  China,  100088} \email{gbl@iapcm.ac.cn}

\thanks{This work is supported in part by NSFC (11471057) and Natural Science Foundation Project of CQ CSTC (cstc2014jcyjA50020).}
 \keywords{}

\begin{abstract}
The hydrodynamic equations with quantum effects are studied in this paper. First we establish the global existence of smooth solutions with small initial data and then in the second part, we establish the convergence of the solutions of the quantum hydrodynamic equations to those of the classical hydrodynamic equations. The energy equation is considered in this paper, which added new difficulties to the energy estimates, especially to the selection of the appropriate Sobolev spaces.
\end{abstract}

\maketitle \numberwithin{equation}{section}
\newtheorem{proposition}{Proposition}[section]
\newtheorem{theorem}{Theorem}[section]
\newtheorem{lemma}[theorem]{Lemma}
\newtheorem{remark}[theorem]{Remark}
\newtheorem{hypothesis}[theorem]{Hypothesis}
\newtheorem{definition}{Definition}[section]
\newtheorem{corollary}{Corollary}[section]
\newtheorem{assumption}{Assumption}[section]

\section{Introduction}
The hydrodynamic equations and related models with quantum effects are extensively studied in recent two decades. In these models, the quantum effects is included into the classical hydrodynamic equations by incorporating the first quantum corrections of $O(\hbar^2)$, where $\hbar$ is the Planck constant. One of the main applications of the quantum hydrodynamic equations is as a simplified but not a simplistic approach for quantum plasmas. In particular, the nonlinear aspects of quantum plasmas of quantum plasmas are much more accessible using a fluid description, in comparasion with kinetic theory. One may see the recent monograph of Haas \cite{Haas11} for many physics backgrounds and mathematical derivation of many interesting models. Many other applications of the quantum hydrodynamic equations consisting of analyzing the flow the electrons in quantum semiconductor devices in nano-size \cite{Gardner94}, where quantum effects like particle tunnelling through potential barriers and built-up in quantum wells, can not be simulated by classical hydrodynamic model. Similar macroscopic quantum models are also used in many other physical fields such as superfluid and superconductivity \cite{Feynman72}.

Let us first consider the following classical hydrodynamic equations in conservation form, describing the motion of the electrons in plasmas by omitting the electric potential
\begin{subequations}\label{equ2}
\begin{numcases}{}
\frac{\partial n}{\partial t}+\frac1m\frac{\partial \Pi_i}{\partial x_i}=0,\label{e2n}\\
\frac{\partial \Pi_j}{\partial t}+\frac{\partial}{\partial x_i}(u_i\Pi_j-P_{ij})=0,\label{e2u}\\
\frac{\partial W}{\partial t}+\frac{\partial}{\partial x_i}(u_iW-u_jP_{ij}+q_i)=0,\label{e2t}
\end{numcases}
\end{subequations}
where $n$ is the density, $m$ is the effective electron mass, $u=(u_1,u_2,u_3)$ is the velocity, $\Pi_j$ is the momentum density, $P_{ij}$ is the stress tensor, $W$ is the energy density and $q$ is the heat flux. In this system, repeated indices are summed over under the Einstein convention. This system also emerges from descriptions of the motion of the electrons in semiconductor devices, with the electrical potential and the relaxation omitted.

As in the classical hydrodynamic equations, the quantum conservation laws have the same form as their classical counterparts. However, to close the moment expansion at the third order, we define the above quantities $\Pi_i,P_{ij}$ and $W$ in terms of the density $n$, the velocity $u$ and the temperature $T$. As usual, the heat flux is assumed to obey the Fourier law $q=-\kappa\nabla T$ and the momentum density is defined by $\Pi_i=mnu_i$, where $m$ is the electron mass and $u$ the velocity. The symmetric stress tensor $P_{ij}$ and the energy density $W$ are defined, with quantum corrections, by
\[
P_{ij}=-nT\delta_{ij}+\frac{\hbar^2n}{12m}\frac{\partial^2}{\partial x_i\partial x_j}\log n+O(\hbar^4)
\]
and
\[
W=\frac{3}{2}nT+\frac12mn|u|^2-\frac{\hbar^2n}{24m}\Delta\log n+O(\hbar^4),
\]
respectively, where $\hbar$ is the Planck constant, and is very small compared to macro quantities.

As far as the quantum corrections are concerned, the quantum correction to the energy density was first derived by Wigner \cite{Wigner32} for thermodynamic equilibrium, and the quantum correction to the stress tensor was proposed by Ancona and Tiersten \cite{AT87} and Ancona and Iafrate \cite{AI89} on the Wigner formalism. See also \cite{Gardner94} for derivation of the system \eqref{equ2}  by a moment expansion of the Wigner-Boltzmann equation and an expansion of the thermal equilibrium Wigner distribution function to $O(\hbar^2)$, leading to the expression for $\Pi$ and $W$ above. We also remark the quantum correction term is closely related to the quantum Bohm potential \cite{Bohm52}
\[
Q(n)=-\frac{\hbar^2}{2m}\frac{\Delta\sqrt{n}}{\sqrt{n}},
\]
where $n$ is the charge density. It relates to the quantum correction term in $P_{ij}$ with
\[
-n\nabla Q(n)=\frac{\hbar^2}{4m}\text{div}(n\nabla^2\log n)=\frac{\hbar^2}{4m}\Delta\nabla\rho -\frac{\hbar^2}{m}\text{div}(\nabla\sqrt{n}\otimes\sqrt{n}).
\]

For the system \eqref{equ2}, there is no dissipation in the second equation. Given $n$ and $T$, the second equation if hyperbolic, and generally we can not expect global smooth solutions for this system. In this paper, we consider the following viscous system by taking into account the stress tensor $\Bbb S$,
\begin{subequations}\label{equ1}
\begin{numcases}{}
\partial_tn+\nabla\cdot(n{u})=0,\\
\partial_t{u}+{u}\cdot\nabla{u}+\frac{1}{mn}\nabla(nT)- \frac{\hbar^2}{12m^2n}\text{div}\{n(\nabla\otimes\nabla)\log n\}=\frac{1}{mn}\text{div}\Bbb S,\\
\partial_tT+{u}\cdot\nabla T+\frac23T\nabla\cdot{u}-\frac2{3n}\nabla\cdot(\kappa\nabla T)+\frac{\hbar^2}{36mn}\nabla\cdot(n\Delta{u})=\frac{2}{3mn}\{\nabla\cdot({u}\Bbb S)-{u}\cdot\text{div}\Bbb S\}.\ \ \ \ \ \ \ \
\end{numcases}
\end{subequations}
Here, $\Bbb S$ is the stress tensor defined by
\[
\Bbb S=\mu(\nabla{u}+(\nabla{u})^T)+\lambda(\text{div}{u})\Bbb I,
\]
where $\Bbb I$ is the $d\times d$ identity matrix, $\mu>0$ and $\lambda$ are the primary coefficients of viscosity and the second coefficients of viscosity, respectively, satisfying $2\mu+3\lambda>0$. Without quantum corrections (i.e., setting $\hbar=0$), this system is exactly the classical hydrodynamic equations studied in the seminal paper of Matsumura and Nishida \cite{MN80}.

Although important, there is little result on the system \eqref{equ1} to the best of our knowledge. But there does exist a large amount of work for system very similar to \eqref{equ1}. These work comes from two main origins. The first one is from the quantum correction to various hydrodynamic equations, especially in semiconductors and in plasmas. Gardner \cite{Gardner94} derived the full 3D quantum hydrodynamic model by a moment expansion of the Wigner-Boltzmann equation. Hsiao and Li \cite{HL09} reviewed the recent progress on well-posedness, stability analysis, and small scaling limits for the (bi-polar) quantum hydrodynamic models, where the interested readers may find many useful references therein. Jungel \cite{Jungel10} proved global existence of weak solutions for the isentropic case. See also \cite{LL05,LM04,LM01}.

The other one, being equally important, emerges from the study of the compressible fluid models of Korteweg type, which are usually used to describe the motion of compressible fluids with capillarity effect of materials. See Korteweg \cite{Korteweg1901} and the pioneering work of Dunn and Serrin \cite{DS85}. The reference list can be very long, and we only mention a few of them. Hattori and Li \cite{HLi94,HLi96} considered the local and global existence of smooth solutions for for the fluid model of Korteweg type for small initial data. Wang and Tan \cite{WT11} studied the optimal decay for the compressible fluid model of Korteweg type. Recently, Bian, Yao and Zhu \cite{BYZ14} studied the global existence of small smooth solutions and the vanishing capillarity limit of this model. Jungel \emph{et al} \cite{JLW14} showed a combined incompressible and vanishing capillarity limit in the barotropic compressible Navier-Stokes equations for smooth solutions.

Almost all of the above mentioned results considered the isothermal case, studying only the continuity equation and the momentum equation, or with electric potential described by a Poisson equation. To the best of our knowledge, there is no mathematical studies for the full quantum hydrodynamic system \eqref{equ2}. The system \eqref{equ1} is itself interesting, since the energy equation also includes the quantum effects through the energy density $W$, which brings new features into this system. This makes it different from the previous known results, to be precisely stated in the following.

The aim of this paper is two fold. On one hand, we show the global existence of smooth solutions for \eqref{equ1} with fixed constant $\hbar>0$ when the initial data is small near the constant stationary solution $(n,u,T)=(1,0,1)$. To be precise, we denote the perturbation by $(\rho,u,\theta)=(n-1,u,T-1)$ and transform the \eqref{equ1} into \eqref{QF}. The result is then stated in Theorem \ref{thm1} for $\eqref{QF}$, where the estimates is stated in terms of the planck constant $\hbar$, and we can see clearly how the quantum corrections affect the estimates. On the other hand, since \eqref{equ1} modifies the classical hydrodynamic equations to a macro-micro level in the sense that it incorporates the (micro) quantum corrections, it is expected that as the Planck constant $\hbar\to0$, the solution of the system \eqref{equ1} converges to that of the classical hydrodynamic equations. This limit is rigorously studied in this paper and stated in Theorem \ref{thm2}. In particular, algebraic convergence rate is given in terms of $\hbar$.

Among others, one of the main novelties is the selection of the Sobolev space like $(\rho, u,\theta)\in H^{k+2}\times H^{k+1}\times H^k$. The underlying reason lies in the fact that the quantum effects in the energy density involves higher order derivatives of the velocity, and hence we cannot seek solutions in the same Sobolev spaces for $\theta$ and $u$.

This paper is organized as follows. In Sect. 2, we present some preliminaries. We translate the system \eqref{equ1} into a convenient form and state the main results in this paper. In Sect. 3, we give the \emph{a priori} estimates, and then prove Theorem \eqref{thm1} (existence result) at the end of this section. Finally, in Sect. 4, we prove Theorem \eqref{thm2}, by showing the convergence of the solutions of the quantum hydrodynamic equations \eqref{QF} to that of the classical hydrodynamic equations \eqref{zero}. The algebraic convergence rate is also given in terms of the Planck constant $\hbar$.

Notations. Throughout, $C$ denotes some generic constant independent of time $t>0$ and the Planck constant $\hbar>0$. Let $p\in[1,\infty]$, $L^p$ denotes the usual Lebesgue space with norm $\|\cdot\|_{L^p}$. When $p=2$, it is usually write $\|\cdot\|=\|\cdot\|_{L^p}$, omitting the subscript. Let $H^k$ denote the Sobolev space of the measurable functions whose generalized derivatives up to $k^{th}$ order belong to $L^2$ with norm $\|\cdot\|_{H^k}=(\sum_{j=0}^k\|D^{j}\cdot\|^2)^{1/2}$. $\dot H^{s,p}$ denotes the homogeneous Sobolev spaces and $[A,B]=AB-BA$ denotes the commutator of $A$ and $B$.

\section{Preliminaries and Main Results}
In this section, we reformulate the system \eqref{equ1} in convenient variables. First we take $(n,u,T)=(1,0,1)$ to be a constant solution to \eqref{equ1} and consider
\[
\rho=n-1,u=u,\theta=T-1.
\]
In these unknowns, with $m=1$, the \eqref{equ1} transforms into
\begin{subequations}\label{QF}
\begin{numcases}{}
\partial_t\rho+u\cdot\nabla\rho+(1+\rho)\text{div}u=0,\label{QF1}\\
\partial_tu-\frac{\mu}{\rho+1}\Delta u-\frac{\mu+\lambda}{\rho+1}\nabla\text{div}u =-u\cdot\nabla u -\nabla\theta-\frac{\theta+1}{\rho+1}\nabla\rho \nonumber\\
\ \ \ \ \ +\frac{\hbar^2}{12}\frac{\Delta\nabla\rho}{\rho+1} -\frac{\hbar^2}{3}\frac{\text{div} (\nabla{\sqrt{\rho+1}} \otimes\nabla{\sqrt{\rho+1}})}{\rho+1},\label{QF2}\\
\partial_t\theta-\frac{2\kappa}{3(1+\rho)}\Delta\theta =-u\cdot\nabla\theta -\frac23(\theta+1)\nabla\cdot u \nonumber\\
\ \ \ \ \ +\frac{\hbar^2}{36(1+\rho)}\text{div}((1+\rho) \Delta u) +\frac{2}{3(1+\rho)}\left\{\frac{\mu}{2}|\nabla u+(\nabla u)^T|^2+\lambda(\text{div} u)^2\right\},\label{QF3}\ \ \ \ \ \ \ \
\end{numcases}
\end{subequations}
with initial data
\[
(\rho,u,\theta)(0,x)=(\rho_0,u_0,\theta_0)(x)=(n_0-1,u_0,T_0-1)(x).
\]
We first state the local-in-time existence of smooth solutions to \eqref{QF}. To be precise, we first set
\begin{equation*}
\begin{split}
|||(\rho,u,\theta)|||_0^2:=&\|(\rho,u,\theta)\|^2_{L^2} +\|\nabla\rho\|^2_{L^2} +\|(\hbar\nabla\rho,\hbar\nabla u)\|^2_{L^2}+\|\hbar^2\Delta\rho\|^2_{L^2},\\
|||(\rho,u,\theta)|||_k^2:=&|||(\rho,u,\theta)|||_{k-1}^2 +|||(\nabla^k\rho,\nabla^ku,\nabla^k\theta)|||_{0}^2,
\end{split}
\end{equation*}
and
\begin{equation*}
\begin{split}
\mathcal E_k(0,T)=\left\{(\rho,u,\theta)(t,\cdot): \sup_{t\in[0,T]}|||(\rho,u,\theta)(t)|||_k^2<\infty\right\}.
\end{split}
\end{equation*}

\begin{theorem}[Local existence]
For any initial data such that $n_0\geq \delta>0$ is satisfied and $(\rho_0=n_0-1,u_0,\theta_0)\in H^{k+2}\times H^{k+1}\times H^{k}$ $(k\geq3)$, there exists some $T>0$ such that the Cauchy problem \eqref{QF} has a unique solution $(\rho,u,\theta)$ in $[0,T]$ such that $(\rho,u,\theta)\in \mathcal E_k(0,T)$ and
$$|||(\rho,u,\theta)(t)|||_k^2\leq C_k|||(\rho_0,u_0,\theta_0)|||_k^2.$$
\end{theorem}
This theorem can be proved in a similar fashion as in \cite{HLi94} by the dual argument and iteration techniques, and hence omitted for brevity.

Now, we consider the global existence of smooth solutions. Let $T>0$, we set
\begin{equation}\label{smallness}
E=\sup_{0\leq t\leq T}{|||}(\rho,u,\theta)(t,\cdot)|||_3.
\end{equation}
One of the main purpose is to show the following
\begin{theorem}[Global existence]\label{thm1}
Suppose the initial data
\[
(\rho,u,\theta)(0)\in H^5\times H^4\times H^3
\]
and set $E_0:=|||(\rho,u,\theta)(0)|||_3<\infty$. There exists some $\hbar_0>0$, $\varepsilon_0>0$, $\nu_0>0$ and $C_0<\infty$, such that if $E_0<\varepsilon_0$ and $\hbar<\hbar_0$, then there exists a unique global in time solution $(\rho,u,\theta)\in \mathcal E_3(0,T)$ of the Cauchy problem \eqref{QF} for any $t\in (0,\infty)$, and the following estimates hold
\begin{equation*}
\begin{split}
{|||}(\rho,u,\theta)(t)|||_3^2 +\nu_0\int_0^t\sum_{k=1}^4\|\nabla^k(\rho,u,\theta, \hbar\nabla u,\hbar\nabla\rho)(s)\|^2ds
\leq C_0{|||}(\rho,u,\theta)(0)|||_3^2,
\end{split}
\end{equation*}
where the constants $\nu_0>0$ and $C_0>0$ are independent of time $t$ and $\hbar$.
\end{theorem}

Formally, as $\hbar\to0$, \eqref{QF} tends to the following classical hydrodynamic equations for $(\rho^0,u^0,\theta^0)$ (studied in \cite{MN80})
\begin{subequations}\label{zero}
\begin{numcases}{}
\partial_t\rho^0+u^0\cdot\nabla\rho^0+(1+\rho^0)\text{div}u^0=0,\label{z1}\\
\partial_tu^0-\frac{\mu}{\rho^0+1}\Delta u^0-\frac{\mu+\lambda}{\rho^0+1}\nabla\text{div}u^0 =-u^0\cdot\nabla u^0 -\nabla\theta^0-\frac{\theta^0+1}{\rho^0+1}\nabla\rho^0,\ \ \ \ \label{z2}\\
\partial_t\theta^0-\frac{2\kappa}{3(1+\rho^0)}\Delta\theta^0 =-u^{0}\cdot\nabla\theta^0 -\frac23(\theta^0+1)\nabla\cdot u^0 \nonumber\\
\ \ \ \ \ \ \ \ \ \ \ \ \ \ \ \ \ \ \ \ \ \ \ \ \ \ \ \ \ \ \ \ +\frac{2}{3(1+\rho^0)}\left\{\frac{\mu}{2}|\nabla u^0+(\nabla u^0)^T|^2+\lambda(\text{div} u^0)^2\right\}.\ \ \ \ \label{z3}\ \ \ \ \ \ \ \
\end{numcases}
\end{subequations}

The convergence result is stated in the following
\begin{theorem}[Semiclassical limit]\label{thm2}
Let $(\rho^{\hbar},u^{\hbar},\theta^{\hbar})$ be the solution of \eqref{QF} and $(\rho^0,u^0,\theta^0)$ be the solution of \eqref{zero} with the same initial data $(\rho_0,u_0,\theta_0)\in H^5\times H^4\times H^3$. Then for all fixed time $T\in(0,\infty)$, we have the algebraic convergence
\begin{equation*}
\begin{split}
\sup_{t\in[0,T]}\|(\rho^{\hbar}-\rho^0,u^{\hbar}-u^0,\theta^{\hbar}-\theta^0)\|_{H^1}^2\leq \left[{c_2}e^{c_1T}/{c_1}\right]\hbar^4
\end{split}
\end{equation*}
and
\begin{equation*}
\begin{split}
\sup_{t\in[0,T]}\|(\rho^{\hbar}-\rho^0,u^{\hbar}-u^0,\theta^{\hbar}-\theta^0)\|_{H^2}^2\leq \left[{c_2}e^{c_1T}/{c_1}\right]\hbar^2,
\end{split}
\end{equation*}
for some constant positive constants $c_1$ and $c_2$, independent of $\hbar$ and $t$.
\end{theorem}

The following three lemmas will be frequently used, and hence cited here for reader's convenience.
\begin{lemma}[Gagliardo-Nirenberg \cite{Nirenberg59}]\label{GN}
Let $p,q,r\in [1,\infty]$ and $0\leq i,j\leq l$ be integers, there exist some generic constants $\theta\in[0,1]$ and $C>0$, such that
\[
\|\nabla^ju\|_{L^p(\Bbb R^N)}\leq C\|\nabla^lu\|_{L^q(\Bbb R^N)}^{\theta}\|\nabla^iu\|_{L^r(\Bbb R^N)}^{1-\theta},
\]
where
\[
j-\frac{N}{p}=\theta(l-\frac{N}{q})+(1-\theta)(i-\frac{N}{r}).
\]
When $\theta=1$, $l-j\neq N/q$.
\end{lemma}

\begin{lemma}\label{lem0}
Let $g(\rho)$ and $g(\rho,\theta)$ be smooth functions of $\rho$ and $(\rho,\theta)$, respectively, with bounded derivatives of any order, and $\|\rho\|_{L^{\infty}}<1$. Then for any integer $m\geq1$, we have
\[
\|\nabla^mg(\rho)\|_{L^p}\leq C\|\nabla^m\rho\|_{L^p},\ \|\nabla^mg(\rho,\theta)\|_{L^p}\leq C\|\nabla^m\rho,\nabla^m\rho\|_{L^p},\ \ \ \forall p\in[1,\infty],
\]
where $C$ may depend on $g$ and $m$. In particular,
\[
\|\partial_x^{\alpha}\left(\frac{\theta+1}{1+\rho}\right)\|_{L^p} \leq C\|(\rho,\theta)\|_{\dot H^{|\alpha|,p}},\ \ \ \forall p\in[1,\infty].
\]
\end{lemma}

\begin{proof}
This can be proved in a similar fashion as in \cite{WT11} and \cite{BYZ14} making use of the Gagliardo-Nirenberg inequality, and hence omitted here for brevity.
\end{proof}

\begin{lemma}[Kato-Ponce \cite{KP88}]\label{Le-inequ}
Let $\alpha$ be any multi-index with $|\alpha|=k$ and $p\in (1,\infty)$. Then there exists some constant $C>0$ such that
\[
\|\partial_x^{\alpha}(fg)\|_{L^p}\leq  C\{\|f\|_{L^{p_1}}\|g\|_{\dot H^{s,p_2}}+\|f\|_{\dot H^{s,p_3}}\|g\|_{L^{p_4}}\},
\]
\[
\|[\partial_x^{\alpha},f]g\|_{L^p}\leq C\{\|\nabla f\|_{L^{p_1}}\|g\|_{\dot H^{k-1,p_2}}+\|f\|_{\dot H^{k,p_3}}\|g\|_{L^{p_4}}\},
\]
where $f,g\in \mathcal{S}$, the Schwartz class and $p_2,p_3\in (1,+\infty)$ such that
\[
\frac1p=\frac1{p_1}+\frac1{p_2}=\frac1{p_3}+\frac1{p_4}.
\]
\end{lemma}

\section{\emph{A priori} estimates}
In this section, we establish useful \emph{a priori} estimates of the solutions to \eqref{QF}. First of all, we let the Planck constant $\hbar<1$. To simplify the proof slightly, we assume that there exists a positive number $\varepsilon\ll1$ such that
\begin{equation}\label{equ9}
E=\sup_{t\in[0,T]}|||(\rho,u,\theta)(t)|||_3\leq \varepsilon,
\end{equation}
which together with Sobolev embedding, implies that
\begin{equation}\label{equ10}
\sup_{t\in[0,T]}\|(\rho,u,\theta),\nabla(\rho,u,\theta),\nabla^2\rho,\hbar\nabla^2(\rho,u), \hbar^2\nabla^3\rho\|_{L^{\infty}}\leq CE\leq C\varepsilon,
\end{equation}
and from \eqref{QF} the following
\begin{equation}\label{equ11}
\|\partial_t\rho\|_{L^{p}}\leq C\|\nabla\cdot(u(1+\rho))\|_{L^{p}}\leq CE\leq C\varepsilon,\ \ \forall 1\leq p\leq \infty,
\end{equation}
and
\begin{equation}\label{equ29}
\begin{split}
\|\partial_t\theta\|_{L^{p}}\leq & \frac{4\kappa}{3}\|\Delta\theta\|_{L^{p}} +\|u\cdot\nabla\theta\|_{L^{p}} +\|\text{div}u\|_{L^{p}} +\hbar^2\|\nabla((1+\rho)\Delta u)\|_{L^{p}}\\
& +C\|\nabla u\|_{L^{2p}}^2 \leq CE\leq C\varepsilon,\ \ \ \forall 1\leq p\leq 6.
\end{split}
\end{equation}
In particular, we choose $\varepsilon$ small enough such that
\begin{equation}\label{equ12}
\sup_{t\in[0,T]}\|(\rho,\theta)(t)\|_{L^{\infty}}\leq 1/2.
\end{equation}
\subsection{Basic estimates}\label{Sect3.1}
Now, we consider the zeroth order estimates for the system \eqref{QF}. As in \cite{MN80}, we set
\begin{equation}\label{equ3}
s=(1+\theta)/(1+\rho)^{2/3}-1,
\end{equation}
and define a function $E^0(\rho,u,s)$ for $\rho,u=(u^1,u^2,u^3)$ and $s$ by
\begin{equation}\label{equ4}
\begin{split}
E^0(\rho,u,s)=&\frac{3R(1+s)}{2}((1+\rho)^{5/3}-1-\frac{5\rho}{3}) +\frac{(1+\rho)}{2}|u|^2+Rs\rho+\frac{3R(1+\rho)s^2}{4}.
\end{split}
\end{equation}

The following lemma is proved in \cite{MN80}.
\begin{lemma}\label{lem1}
There exists constants $0<\rho_2\leq 1/2$ and $0<C_1\leq C_2<\infty$ such that $E^0$ is positive definite, i.e.,
\begin{equation*}
\rho^2+|u|^2+\theta^2\leq C_1E^0\leq C_2(\rho^2+|u|^2+\theta^2),\ \ \ \ for\ |\rho|\leq \rho_2.
\end{equation*}
\end{lemma}

We first prove the zeroth order estimates in the following
\begin{proposition}\label{prop1}
There exists a constant $\varepsilon_0>0$ such that if $E\leq \varepsilon\leq \varepsilon_0$, then the following a priori estimates holds for all $t\in[0,T]$,
\begin{equation}\label{equ42}
\begin{split}
{|||}(\rho,u,\theta)(t)|||_0^2 +\nu_0\int_0^t\|D(\rho,u,\theta)(s),(\hbar\Delta u,\hbar\Delta\rho)(s)\|^2ds \leq C{|||}(\rho,u,\theta)(0)|||_0^2,
\end{split}
\end{equation}
where $\nu_0>0$, $C=C(\varepsilon_0)$ are independent of $t$.
\end{proposition}
The proof if postponed to the end of Section \ref{Sect3.1}.

\begin{lemma}\label{lem2}
There exists $0<\varepsilon_0<1$ and $h_0>0$ such that if $E\leq \varepsilon\leq \varepsilon_0$ and $\hbar\leq h_0$, then for a suitable $\beta>0$, there holds
\begin{equation}\label{equ26}
\begin{split}
\|(\rho, u, s)(t)\|^2& +\beta\|\nabla\rho(t)\|^2 +\nu_0\int_0^t\|\nabla(\rho,u,s)(\tau)\|^2d\tau +\nu_0\int_0^t\|\hbar\Delta\rho(\tau)\|^2d\tau \\
\leq & C{|||}(\rho,u,\theta)(0)|||_0^2 +{C\hbar^4}\int_0^t\|\Delta u\|^2ds,
\end{split}
\end{equation}
for some constant $C>0$ independent of $t$.
\end{lemma}
\begin{proof}
Under the transform of \eqref{equ3}, the system \eqref{QF} is transformed into the following system for $(\rho,u,s)$

\begin{subequations}\label{QF'}
\begin{numcases}{}
\partial_t\rho+u\cdot\nabla\rho+(1+\rho)\text{div}u=0,\label{QF'1}\\
\partial_tu+\frac{\nabla((1+\rho)^{\gamma}(1+s))}{1+\rho} -\frac{\mu}{\rho+1}\Delta u-\frac{\mu+\lambda}{\rho+1}\nabla\text{div}u +u\cdot\nabla u \nonumber\\
\ \ \ \ \ \ \ \ =\frac{\hbar^2}{12}\frac{\Delta\nabla\rho}{\rho+1} -\frac{\hbar^2}{3}\frac{\text{div} (\nabla{\sqrt{\rho+1}}\otimes\nabla{\sqrt{\rho+1}})} {\rho+1}=:g_{\hbar}\\
\partial_ts+u\cdot\nabla\theta -\frac{2\kappa}{3(1+\rho)}\text{div}\left\{\frac{\nabla s}{1+\rho}+\frac{2(1+s)\nabla\rho}{3(1+\rho)^2}\right\} \nonumber\\ \ \ \ \ \ \ \ \ \ \ \ -\kappa\gamma(\gamma-1)\left\{\frac{\nabla s}{(1+\rho)^2}+\frac{2(1+s)}{3(1+\rho^3)}\nabla\rho\right\} \nabla\rho \nonumber\\
\ \ \ \ \ \ \ \ \ \ \ +\frac{2}{3(1+\rho)^{\gamma}} \left\{\frac{\mu}{2}|\nabla{u}+(\nabla{u})^T|^2 +\lambda(\text{div}{u})^2\right\}\\
\ \ \ \ \ \ \ \ =-\frac{\hbar^2}{36(1+\rho)^{\gamma}} \text{div}((1+\rho)\Delta{u})=:h_{\hbar}.
\end{numcases}
\end{subequations}

Recall that $E^0(\rho,u,s)$ is given in \eqref{equ4}. We compute
\begin{equation}\label{equ5}
\begin{split}
\partial_tE^0=&(1+\rho)u\cdot u_t +\{\frac{u^2}{2} +\frac{\gamma}{\gamma-1}(1+s((1+\rho)^{\gamma-1}-1))+s \\
&+\frac{1}{2(\gamma-1)}s^2\}\rho_t +\{\frac{1}{\gamma-1}((1+\rho)^{\gamma}-1-\gamma\rho) +\frac{(1+\rho)s}{\gamma-1}\}s_t\\
=&\text{div}\{\cdots\}-\mu|\nabla u|^2-(\mu+\lambda)|\nabla\cdot u|^2 -2\kappa(\gamma-1)\nabla s\cdot\nabla\rho\\
&-\kappa(\gamma-1)^2|\nabla\rho|^2-\kappa|\nabla s|^2+O(E)|D(\rho,u,s)|^2\\
&+\underbrace{(1+\rho)u\cdot g_{\hbar}}_{I}-\underbrace{\{\frac{1}{\gamma-1}((1+\rho)^{\gamma}-1-\gamma\rho) +\frac{(1+\rho)s}{\gamma-1}\}h_{\hbar}}_{II}.
\end{split}
\end{equation}
Now, we consider the integration in space of the last two terms $I$ and $II$ on the RHS of \eqref{equ5}. For the first term $I$, by integration by parts, and using \eqref{QF'1}, we obtain
\begin{equation}\label{equ40}
\begin{split}
\int Idx=&-\frac{\hbar^2}{12}\int\nabla\cdot u\Delta\rho +\frac{\hbar^2}{12}\int\nabla u(\nabla\rho\otimes\nabla\rho/(1+\rho)).
\end{split}
\end{equation}
The last term on the RHS is easy to be bounded by
\[
\left|\frac{\hbar^2}{12}\int\nabla u(\nabla\rho\otimes\nabla\rho/(1+\rho))\right|\leq C\|\nabla u\|_{L^{\infty}}\|\nabla\rho\|^2\leq C\hbar^2E\|\nabla\rho\|^2.
\]
For the first term on the RHS, we use \eqref{QF1} to obtain
\begin{equation*}
\begin{split}
-\frac{\hbar^2}{12}\int\nabla\cdot u\Delta^2\rho =& \frac{\hbar^2}{12}\int\nabla\nabla\cdot u\nabla\rho\\
= &-\frac{\hbar^2}{12}\int\frac{1}{1+\rho}\nabla(\rho_t+u\cdot\nabla\rho)\nabla\rho\\
=&-\frac{\hbar^2}{24}\frac{d}{dt}\int\frac{|\nabla\rho|^2}{1+\rho} -\frac{\hbar^2}{24}\int\frac{\rho_t|\nabla\rho|^2}{(1+\rho)^2} -\frac{\hbar^2}{24}\int\frac{\nabla(u\cdot\nabla\rho)}{1+\rho}\nabla\rho.
\end{split}
\end{equation*}
But from \eqref{QF1}, it is easy to know that
\[
\|\rho_t\|_{L^{\infty}}\leq C\|\text{div}(\rho u)\|_{L^{\infty}}\leq CE^2 \leq CE,
\]
and by integration by parts,
\begin{equation*}
\begin{split}
\int\frac{\nabla(u\cdot\nabla\rho)}{1+\rho}\nabla\rho =& -\int\frac{\partial_iu\cdot\nabla\rho}{1+\rho}\partial_i\rho -\int\frac{u_i\cdot\partial_i\nabla\rho}{1+\rho}\nabla\rho\\
=& -\int\frac{\partial_iu\cdot\nabla\rho}{1+\rho}\partial_i\rho +\int\frac{\nabla\cdot u |\nabla\rho|^2}{1+\rho} -\int\frac{u\cdot\nabla\rho}{(1+\rho)^2}|\nabla\rho|^2\\
\leq & C(\|\nabla u\|_{L^{\infty}}+\|u\|_{\infty}\|\nabla\rho\|_{L^{\infty}})\|\nabla\rho\|^2\\
\leq &CE\|\nabla\rho\|^2.
\end{split}
\end{equation*}
Therefore it is easy to see from \eqref{equ40} that
\begin{equation*}
\begin{split}
\int Idx \geq & -\frac{\hbar^2}{24}\frac{d}{dt} \int\frac{|\nabla\rho|^2}{1+\rho} -C\hbar^2E\|D(\rho,u,s)\|^2.
\end{split}
\end{equation*}
On the other hand, for the second term $II$ in \eqref{equ5}, we have
\begin{equation*}
\begin{split}
\int IIdx\leq \delta_0\|(\nabla\rho,\nabla s)\|+\frac{C\hbar^4}{\delta_0}\|\Delta u\|^2.
\end{split}
\end{equation*}

In addition to \eqref{equ5}, we compute
\begin{equation}\label{equ6}
\begin{split}
&\frac{\partial}{\partial t}\left\{\frac12|\nabla\rho|^2 +\frac{(1+\rho)^2}{2\mu+\lambda}\nabla\rho\cdot u\right\}\\
=&\left\{\nabla\rho+\frac{(1+\rho)^2} {2\mu+\lambda}u\right\}\cdot\nabla\rho_t +\frac{(1+\rho)^2}{2\mu+\lambda}\nabla\rho\cdot u_t+\frac{2(1+\rho)}{2\mu+\lambda}u \cdot\nabla\rho\rho_t\\
=& \text{div}\{\cdots\}-\frac{1}{2\mu+\lambda} \nabla\rho\cdot\nabla s -\frac{\gamma}{2\mu+\lambda}|\nabla\rho|^2\\
&+\frac{(\text{div}u)^2}{2\mu+\lambda} +O(E)|D(\rho,u,s)|^2 +\underbrace{\frac{(1+\rho)^2}{2\mu+\lambda}\nabla\rho\cdot g_{\hbar}}_{J}.
\end{split}
\end{equation}
After integration in space we obtain for the last term,
\begin{equation*}
\begin{split}
\int Jdx=&\frac{\hbar^2}{12} \int\frac{(1+\rho)}{2\mu+\lambda}\nabla\rho\cdot \Delta\nabla\rho +\frac{\hbar^2}{3} \frac{(1+\rho)}{2\mu+\lambda}\nabla\rho\cdot \text{div}(\nabla\sqrt{\rho+1} \otimes\nabla\sqrt{\rho+1})\\
=& :J_1+J_2.
\end{split}
\end{equation*}
For the term $J_1$, we have
\begin{equation*}
\begin{split}
J_1=&-\frac{\hbar^2}{12} \int\frac{(1+\rho)}{2\mu+\lambda}|\Delta\rho|^2 -\frac{\hbar^2}{12} \int\frac{(1+\rho)}{2\mu+\lambda}|\nabla\rho|^2\Delta\rho \\
\geq & -\frac{\hbar^2}{24} \int\frac{(1+\rho)}{2\mu+\lambda}|\Delta\rho|^2-C\hbar^2E^2\|\nabla\rho\|^2.
\end{split}
\end{equation*}
For the term $J_2$, we have
\begin{equation*}
\begin{split}
J_2\leq & \frac{\hbar^2}{48} \int\frac{(1+\rho)}{2\mu+\lambda}|\Delta\rho|^2+C\hbar^2E^2\|\nabla\rho\|^2.
\end{split}
\end{equation*}

Multiplying \eqref{equ6} with a constant $\beta$ and integration in space, and then add the resultant to \eqref{equ5} integrated in space, we obtain
\begin{equation}\label{equ7}
\begin{split}
\frac{d}{dt}&\int E^0(\rho,u,s)+\frac{\hbar^2}{24}\frac{|\nabla\rho|^2}{1+\rho}+\beta\left(\frac12|\nabla\rho|^2 +\frac{(1+\rho)^2}{2\mu+\lambda}\nabla\rho\cdot u\right)dx\\
&+\int\mu|\nabla u|^2+(\mu+\lambda)|\nabla\cdot u|^2 +\frac43\kappa\nabla s\cdot\nabla\rho +\frac49\kappa|\nabla\rho|^2+\kappa|\nabla s|^2dx\\
&+\beta\int\frac{\nabla\rho\cdot\nabla s}{2\mu+\lambda} +\frac53\frac{1}{2\mu+\lambda}|\nabla\rho|^2 -\frac{(\text{div}u)^2}{2\mu+\lambda} +\frac{\hbar^2}{48}\frac{1+\rho}{2\mu+\lambda}|\Delta\rho|^2dx\\
\leq & O(E)\|D(\rho,u,s)\|^2+\delta_0\|(\nabla\rho,\nabla s)\|^2 +\frac{C\hbar^4}{\delta_0}\|\Delta u\|^2.
\end{split}
\end{equation}

Note that as in \cite{MN80}, if we take $\beta$ small such that
\begin{equation*}
0<\beta<\min\left\{\frac{(2\mu+\lambda)^2}{8(1+\rho_2)^4},(\mu+\lambda)(2\mu+\lambda), {4\kappa(2\mu+\lambda)}\right\},
\end{equation*}
where $\rho_2$ is given in Lemma \ref{lem1}, then
\begin{equation*}
\begin{split} E^0(\rho,u,s) +\beta\left(\frac12|\nabla\rho|^2 +\frac{(1+\rho)^2}{2\mu+\lambda}\nabla\rho\cdot u\right)\geq \frac18(\rho^2+\frac{7}{9}s^2+|u|^2)+\frac\beta4|\nabla\rho|^2,
\end{split}
\end{equation*}
and
\begin{equation*}
\begin{split}
&\mu|\nabla u|^2+(\mu+\lambda)|\nabla\cdot u|^2 +\frac43\kappa\nabla s\cdot\nabla\rho +\frac49\kappa|\nabla\rho|^2+\kappa|\nabla s|^2\\
&+\beta\left(\frac{\nabla\rho\cdot\nabla s}{2\mu+\lambda} +\frac53\frac{1}{2\mu+\lambda}|\nabla\rho|^2 -\frac{(\text{div}u)^2}{2\mu+\lambda}\right)\geq \mu|\nabla u|^2+\frac{5\kappa}{27}|\nabla s|^2+\frac53\kappa|\nabla\rho|^2.
\end{split}
\end{equation*}
Integrating in time over $[0,t]$ and taking $\delta_0$ and $E$ sufficiently small (say, $\delta_0=1/{20}$), we obtain
\begin{equation}\label{equ21}
\begin{split}
\int_{\Bbb R^3}&\left\{\frac18(\rho^2+\frac{\hbar^2}{36}|\nabla\rho|^2+\frac{7}{9}s^2+|u|^2) +\frac\beta4|\nabla\rho|^2\right\}dx\\
&+\frac12\int_0^t\int_{\Bbb R^3}\left\{\mu|\nabla u|^2+\frac{5\kappa}{27}|\nabla s|^2+\frac53\kappa|\nabla\rho|^2 +\frac{\hbar^2}{96}\frac{1}{2\mu+\lambda}|\Delta\rho|^2\right\}\\
\leq & \left.\int_{\Bbb R^3}E^0(\rho,u,s)+\beta\left(\frac12|\nabla\rho|^2 +\frac{(1+\rho)^2}{2\mu+\lambda}\nabla\rho\cdot u\right)dx\right|_{t=0} +{C\hbar^4}\int_0^t\|\Delta u\|^2ds\\
\leq & C{|||}(\rho,u,\theta)(0)|||_0^2 +{C\hbar^4}\int_0^t\|\Delta u\|^2ds.
\end{split}
\end{equation}
Now, properly choose the constant $\nu_0$ and $C>0$, we finish the proof.
\end{proof}

%
%

\begin{lemma}\label{lem3}
Under the same condition in Lemma \ref{lem2}, we have
\begin{equation}\label{equ20}
\begin{split}
& \hbar^2\left\{\|\nabla u(t)\|^2 +\hbar^2\|\Delta\rho(t)\|^2\right\} +\int_0^t\|\hbar\Delta u(s),\hbar\nabla\nabla\cdot u(s)\|^2ds\\
\leq & C\hbar^2\|\nabla u_0,\hbar\Delta\rho_0\|^2 +C\hbar^2E\int_0^t\|\nabla\rho\|^2ds +\delta_1\int_0^t\|\nabla\theta,\nabla\rho\|^2ds +\frac{C\hbar^4}{\delta_1}\int_0^t\|\Delta u\|^2ds,
\end{split}
\end{equation}
for any positive constants $\delta_1>0$ and $t>0$.
\end{lemma}
\begin{proof}
We now take the inner product of \eqref{QF2} with $-\hbar^2\Delta u$ to obtain
\begin{equation}\label{equ8}
\begin{split}
\sum_{i=1}^3L_i:=& \frac{\hbar^2}{2}\frac{d}{dt}\int|\nabla u|^2 +\mu\hbar^2\int\frac{|\Delta u|^2}{1+\rho} +(\mu+\lambda)\hbar^2\int\frac{(\nabla\nabla\cdot u)}{1+\rho}\cdot\Delta u\\
=& \hbar^2\int(u\cdot\nabla u)\Delta u +\hbar^2\int\nabla\theta\cdot\Delta u +\hbar^2\int\frac{\theta+1}{\rho+1}\nabla\rho\cdot\Delta u -\frac{\hbar^4}{12}\int\frac{\Delta\nabla\rho}{1+\rho}\Delta u \\
&+\frac{\hbar^4}{3} \int\frac{\text{div}(\nabla\sqrt{\rho+1} \otimes\nabla\sqrt{\rho+1})}{1+\rho}\Delta u =:\sum_{i=1}^5R_i
\end{split}
\end{equation}
For $L_3$, we have by integration by parts twice
\begin{equation*}
\begin{split}
L_3=&(\mu+\lambda)\hbar^2\int\frac{|\nabla\nabla\cdot u|^2}{(1+\rho)}\\
&-(\mu+\lambda)\hbar^2\int\frac{\partial_i\rho} {(1+\rho)^2}\partial_k\nabla\cdot u\partial_ku^i +(\mu+\lambda)\hbar^2\int\frac{\partial_k\rho} {(1+\rho)^2}\partial_i\nabla\cdot u\partial_ku^i\\
\geq & (\mu+\lambda)\hbar^2\int\frac{|\nabla\nabla\cdot u|^2}{(1+\rho)}-C\hbar^2 E\|\nabla\rho,\Delta u\|^2,
\end{split}
\end{equation*}
since $\|\nabla u\|_{L^{\infty}}\leq C\|u\|_{H^3}\leq CE$.

For $R_1$, after integration by parts twice, we obtain
\begin{equation*}
\begin{split}
R_1=-\hbar^2\int(u\cdot\nabla u)\Delta u-2\hbar^2\int\partial_ku^i\partial_ku^j\partial_iu^j +\hbar^2\int\nabla\cdot u|\nabla u|^2,
\end{split}
\end{equation*}
which implies that
\begin{equation*}
\begin{split}
R_1=&-\hbar^2\int\partial_ku^i\partial_ku^j\partial_iu^j +\frac{\hbar^2}2\int\nabla\cdot u|\nabla u|^2 \leq C\hbar^2E\|\nabla u\|^2.
\end{split}
\end{equation*}
For $R_2$ and $R_3$, by H\"older inequality we obtain
\begin{equation*}
\begin{split}
R_2+R_3\leq \delta_1(\|\nabla\theta,\nabla\rho\|^2 +\frac{C\hbar^4}{\delta_1}\|\Delta u\|^2,\ \ \ \forall\delta_1>0.
\end{split}
\end{equation*}
For $R_4$, we have by integration by parts and \eqref{QF1}
\begin{equation*}
\begin{split}
R_4=& -\frac{\hbar^4}{12}\int\frac{\Delta\rho\nabla\rho}{(1+\rho)^2}\cdot\Delta u +\frac{\hbar^4}{12}\int\frac{\Delta\rho}{1+\rho}\Delta\nabla\cdot u\\
=&-\frac{\hbar^4}{12}\int\frac{\Delta\rho\nabla\rho}{(1+\rho)^2}\cdot\Delta u \\
& -\frac{\hbar^4}{12}\int\frac{\Delta\rho}{(1+\rho)^2}\{\partial_t\Delta\rho +[\Delta,1+\rho]\text{div}u+[\Delta,u]\nabla\rho+u\cdot\nabla\Delta\rho\} =\sum_{i=1}^5 R_{4i}.
\end{split}
\end{equation*}
It is easy to show the following estimates
\begin{equation*}
\begin{split}
R_{41}\leq C\hbar^4\|\nabla\rho\|_{L^{\infty}}\|\Delta\rho,\Delta u\|^2 \leq CE\hbar^2\|\hbar\Delta\rho,\hbar\Delta u\|^2,
\end{split}
\end{equation*}
\begin{equation*}
\begin{split}
R_{42}=& -\frac{\hbar^4}{24}\frac{d}{dt}\int\frac{|\Delta\rho|^2}{(1+\rho)^2} -\frac{\hbar^4}{12}\int\frac{\partial_t\rho}{(1+\rho)^3}|\Delta\rho|^2\ \ \ \ \ \\
\leq & -\frac{\hbar^4}{24}\frac{d}{dt}\int\frac{|\Delta\rho|^2}{(1+\rho)^2} +CE\hbar^2\|\hbar\Delta\rho\|^2,
\end{split}
\end{equation*}
thanks to \eqref{equ11} and
\begin{equation*}
\begin{split}
R_{43}+R_{44} \leq & C\hbar^4\|\Delta\rho\|(\|[\Delta,1+\rho]\text{div}u\| +\|[\Delta,u]\nabla\rho\|)\\
\leq & C\hbar^4\|\Delta\rho\|(\|\Delta u\|\|\nabla\rho\|_{L^{\infty}}+\|\Delta\rho\|\|\nabla u\|_{L^{\infty}})\\
\leq & CE\hbar^2\|\hbar\Delta\rho,\hbar\Delta u\|^2,
\end{split}
\end{equation*}
thanks to Lemma \ref{Le-inequ}, and by integration by parts
\begin{equation*}
\begin{split}
R_{45}= & \frac{\hbar^4}{12}\int\frac{\nabla\Delta\rho}{(1+\rho)^2}\cdot u\Delta\rho +\frac{\hbar^4}{12}\int\frac{|\Delta\rho|^2}{(1+\rho)^2}\text{div}u -\frac{\hbar^4}{6}\int\frac{|\Delta\rho|^2}{(1+\rho)^3}u\cdot\nabla\rho\\
= & \frac{\hbar^4}{24}\int\frac{|\Delta\rho|^2}{(1+\rho)^2}\text{div}u -\frac{\hbar^4}{4}\int\frac{|\Delta\rho|^2}{(1+\rho)^3}u\cdot\nabla\rho \\
\leq & CE\hbar^2\|\hbar\Delta\rho\|^2.
\end{split}
\end{equation*}
Therefore, we obtain
\begin{equation*}
\begin{split}
R_{4}\leq & -\frac{\hbar^4}{24}\frac{d}{dt}\int\frac{|\Delta\rho|^2}{(1+\rho)^2} +CE\hbar^2\|\hbar\Delta\rho\|^2.
\end{split}
\end{equation*}
For the term $R_5$, it is easy to show that
\begin{equation*}
\begin{split}
R_{5}\leq & CE\hbar^2\|\hbar\Delta\rho\|^2.
\end{split}
\end{equation*}
Hence, putting all the estimates together, we have from \eqref{equ8} that
\begin{equation*}
\begin{split}
& \frac{\hbar^2}{2}\frac{d}{dt}\int|\nabla u|^2 +\frac{\hbar^4}{24}\frac{d}{dt}\int\frac{|\Delta\rho|^2}{(1+\rho)^2} +\mu\hbar^2\int\frac{|\Delta u|^2}{1+\rho} +(\mu+\lambda)\hbar^2\int\frac{|\nabla\nabla\cdot u|^2}{(1+\rho)}\\
\leq & C\hbar^2 E\|\nabla\rho,\Delta u\|^2 +\delta_1\|\nabla(\rho,\theta)\|^2 +\frac{C\hbar^4}{\delta_1}\|\Delta u\|^2 +CE\hbar^2\|\hbar\Delta\rho,\hbar\Delta u\|^2.
\end{split}
\end{equation*}
Take $\varepsilon_0$ and $\hbar_0$ small, then for any $\varepsilon\leq \varepsilon_0$ and $\hbar\leq \hbar_0$, integration in time over $[0,t]$ yields the result for any positive constant $\delta_1>0$.
\end{proof}

\begin{proof}[Proof of Proposition of \ref{prop1}] By \eqref{equ3} and Lemma \ref{lem1}, the left hand side of \eqref{equ26} is equivalent to the norm
\begin{equation*}
\|(\rho,u,\theta)(t)\|^2+\beta\|\nabla\rho(t)\|^2+\nu_0\int_0^t\|(\nabla\rho,\nabla u,\nabla\theta,\hbar\Delta\rho)(s)\|^2ds,
\end{equation*}
for some another suitable constant $\nu_0>0$. Hence \eqref{equ26} implies that
\begin{equation}\label{equ22}
\begin{split}
&\|(\rho,u,\theta)(t)\|^2+\beta\|\nabla\rho(t)\|^2+\nu_0\int_0^t\|(\nabla\rho,\nabla u,\nabla\theta,\hbar\Delta\rho)(s)\|^2ds\\
\leq & C{|||}(\rho,u,\theta)(0)|||_0^2 +{C\hbar^4}\int_0^t\|\Delta u\|^2ds.
\end{split}
\end{equation}
Now, taking $\delta_1$ and $\hbar_0$ in \eqref{equ20} sufficiently small, say, $\delta_1=\nu_0/4$ and $\hbar_0^2<{\nu_0}/{4C\varepsilon_0}$, we then obtain
\begin{equation}\label{equ23}
\begin{split}
& \hbar^2\left\{\|\nabla u(t)\|^2 +\hbar^2\|\Delta\rho(t)\|^2\right\} +\int_0^t\|(\hbar\Delta u,\hbar\nabla\nabla\cdot u)(s)\|^2ds\\
\leq & C\hbar^2\|\nabla u_0,\hbar\Delta\rho_0\|^2 +\frac{\nu_0}{2}\int_0^t\|(\nabla\rho,\nabla\theta)\|^2ds +\frac{4C\hbar^4}{\nu_0}\int_0^t\|\Delta u\|^2ds.
\end{split}
\end{equation}
Add \eqref{equ22} and \eqref{equ23} together, and then taking $\hbar_0$ even smaller such that $C\hbar_0^2+4C\hbar_0^2/\nu_0<1/2$, we then obtain
\begin{equation}\label{equ25}
\begin{split}
{|||}(\rho,u,\theta)(t)|||^2+ \nu_0\int_0^t\|(\nabla\rho,\nabla u,\nabla\theta,\hbar\Delta\rho,\hbar\Delta u)(s)\|^2ds\leq C{|||}(\rho,u,\theta)(0)|||_0^2,
\end{split}
\end{equation}
for some positive constant $\nu_0>0$ depends only on $\mu,\lambda$ and $\kappa$. in particular, $\nu_0$ and $C$ are both independent of $t$.
\end{proof}

\subsection{Higher order estimates}
In the following, we denote $\partial^{\alpha}=\partial_{x_1}^{\alpha_1} \partial_{x_2}^{\alpha_2}\partial_{x_3}^{\alpha_3}$ the partial differential derivative operator with multi-index $\alpha=(\alpha_1,\alpha_2,\alpha_3)$. For our purpose, $|\alpha|\leq 3$ suffices. We sometimes abuse the notation to use $\alpha\pm1$ to stand for $\alpha\pm\beta$ for a multi-index with $|\beta|=1$ and $\alpha\geq\beta$ in the case of $\alpha-\beta$. We will prove the following
\begin{proposition}\label{prop3}
Let $\alpha$ be any multi-index with $1\leq |\alpha|\leq 3$ and $s=|\alpha|$. There exist some constants $\varepsilon_0>0$ and $\hbar_0>0$ such that if $E\leq \varepsilon_0$ and $\hbar\leq\hbar_0$, then the following \emph{a priori} estimates hold for all $t\in[0,T]$,
\begin{equation}\label{equ43}
\begin{split}
&{|||}\partial^{\alpha}(\rho,u,\theta)(t)|||^2 +\nu_0\int_0^t\|\partial^{\alpha}\nabla(\rho,u,\theta,\hbar\nabla\rho,\hbar\nabla u)(\tau)\|^2d\tau \\
\leq & C{|||}\partial^{\alpha}(\rho,u,\theta)(0)|||^2
+C\int_0^t\|\nabla(\rho,u,\theta)\|^2_{H^{s-1}}d\tau +C\int_0^t\|\hbar\Delta(\rho,u)(\tau)\|_{H^{s-1}}^2d\tau,
\end{split}
\end{equation}
for some $\nu_0>0$ and $C=C(\varepsilon_0)$ independent of $t$.
\end{proposition}
This proposition is proved as a direct sequence of the following lemmas.
\begin{lemma}\label{lem8}
Under the assumptions in Proposition \ref{prop3}, there exists some constants $\delta_0<1$ and $\varepsilon_0<1$ sufficiently small, such that
\begin{equation}\label{equ27}
\begin{split}
&\|\partial^{\alpha}\theta(t)\|^2 +\kappa\int_0^t\|\partial^{\alpha}\nabla\theta(\tau)\|^2d\tau\\
\leq & \|\partial^{\alpha}\theta_0\|^2 +\delta_0(\mu+\lambda)\int_0^t\|\partial^{\alpha}\nabla\cdot u(\tau)\|^2d\tau +\frac{C\hbar^4}{\delta_0\kappa}\int_0^t\|\partial^{\alpha}\Delta u(\tau)\|^2 d\tau \\
&+\frac{C}{\delta_0}\int_0^t\|\nabla(\rho,u,\theta)(\tau)\|_{\dot H^{s-1}}^2d\tau +\frac{CE^2\hbar^2}{\delta_0\kappa}\int_0^t\|\Delta u(\tau)\|_{\dot H^{s-1}}^2d\tau,
\end{split}
\end{equation}
for all $\delta\leq \delta_0$ and $E\leq \varepsilon\leq \varepsilon_0$, where $C$ is independent of $t$.
\end{lemma}
\begin{proof}
Applying $\partial^{\alpha}$ to \eqref{QF3} and then taking inner product of the resultant with $\partial^{\alpha}\theta$ to obtain
\begin{equation*}
\begin{split}
&\frac12\frac{d}{dt}\|\partial^{\alpha}\theta\|^2 +\frac{2\kappa}{3}\int\frac{|\partial^{\alpha}\nabla\theta|^2}{(1+\rho)} =\frac{2\kappa}{3}\int\frac{\nabla\partial^{\alpha}\theta\cdot\nabla\rho}{(1+\rho)^2} \partial^{\alpha}\theta\\
-&\int\partial^{\alpha}(u\cdot\nabla\theta) \partial^{\alpha}\theta -\frac23\int\partial^{\alpha}\left((\theta+1)\nabla\cdot u\right) \partial^{\alpha}\theta +\frac{\hbar^2}{36}\int\partial^{\alpha}\left(\frac{\text{div}((1+\rho) \Delta u)}{(1+\rho)}\right)\partial^{\alpha}\theta \\
&+\frac{2}{3}\int\partial^{\alpha}\left(\frac{\mu|\nabla u+(\nabla u)^T|^2+2\lambda(\text{div} u)^2}{(1+\rho)}\right)\partial^{\alpha}\theta =\sum R_i.
\end{split}
\end{equation*}
For the first term $R_1$, we have
\begin{equation*}
\begin{split}
R_1\leq & C\kappa\|\partial^{\alpha}\nabla\theta\| \|\nabla\rho\|_{L^3}\|\partial^{\alpha}\theta\|_{L^6}
\leq C\kappa E\|\partial^{\alpha}\nabla\theta\|^2.
\end{split}
\end{equation*}
For the term $R_2$, we have
\begin{equation*}
\begin{split}
R_2\leq & \|\partial^{\alpha}(u\cdot\nabla\theta)\|_{L^{6/5}}\|\partial^{\alpha}\theta\|_{L^6}\\
\leq &(\|u\|_{L^3}\|\partial^{\alpha}\nabla\theta\|_{L^2} +\|\partial^{\alpha}u\|_{L^6}\|\nabla\theta\|_{L^{3/2}}) \|\nabla\partial^{\alpha}\theta\|_{L^2} \\
\leq & \delta_0\kappa\|\nabla\partial^{\alpha}\theta\|_{L^2}^2 +\frac{CE^2}{\delta_0\kappa}\|\nabla\partial^{\alpha}u\|_{L^2}^2,
\end{split}
\end{equation*}
thanks to Lemma \ref{Le-inequ}. For the term $R_3$, since $|\theta+1|\leq 3/2$ by \eqref{equ11}, we have
\begin{equation*}
\begin{split}
R_3\leq &(\|\theta+1\|_{L^{\infty}}\|\partial^{\alpha}\nabla\cdot u\|_{L^2}+\|\partial^{\alpha}\theta\|_{L^{6}}\|\nabla\cdot u\|_{L^3})\|\partial^{\alpha}\theta\|_{L^2}\\
\leq & \delta_0\kappa\|\nabla\partial^{\alpha}\theta\|_{L^2}^2 +\delta_0(\mu+\lambda)\|\partial^{\alpha}\nabla\cdot u\|_{L^2}^2 +\frac{C(1+E^2)}{\delta_0}\|\partial^{\alpha}\theta\|_{L^2}^2.
\end{split}
\end{equation*}
For the term $R_4$, by integration by parts,
\begin{equation*}
\begin{split}
R_4=& -\frac{\hbar^2}{36}\int\partial^{\alpha-1}\left(\frac{\text{div}((1+\rho) \Delta u)}{(1+\rho)}\right)\partial^{\alpha+1}\theta\\
\leq & C\hbar^2\|\partial^{\alpha+1}\theta\|\left\{\|\partial^{\alpha-1}\text{div}((1+\rho) \Delta u)\| +\|[\partial^{\alpha-1},\frac{1}{1+\rho}]\text{div}((1+\rho)\Delta u)\|\right\}\\
\leq & C\hbar^2\|\partial^{\alpha+1}\theta\| \left\{\|\partial^{\alpha-1}\text{div}\Delta u\|+\|\Delta u\|_{L^{\infty}}\|\partial^{\alpha-1}\nabla\rho\|_{L^2}\right\} +C\hbar^2\|\partial^{\alpha+1}\theta\|\times\\ &\left\{\|\nabla(\frac{1}{1+\rho})\|_{L^{\infty}}\|\partial^{\alpha-2} \text{div}((1+\rho)\Delta u)\|_{L^2} +\|\frac{1}{1+\rho}\|_{\dot H^{s-1,6}}\|\text{div}((1+\rho)\Delta u)\|_{L^3}\right\},
\end{split}
\end{equation*}
where $s=|\alpha|$. Then making use of Lemma \ref{GN}, \ref{lem0} and \ref{Le-inequ} and \eqref{equ9}-\eqref{equ12}, one obtains
\begin{equation*}
\begin{split}
R_4 \leq & C\hbar^2\|\partial^{\alpha+1}\theta\| \left(\|\Delta u\|_{\dot H^{s}}+\|\Delta u\|_{L^{\infty}}\|\rho\|_{\dot H^s}\right) \\
& +C\hbar^2\|\partial^{\alpha+1}\theta\|\|\nabla\rho\|_{L^{\infty}} (\|\Delta u\|_{\dot H^{s-1}} +\|\rho\|_{\dot H^{s-1,6}}\|\Delta u\|_{L^{3}})\\
&+C\hbar^2\|\partial^{\alpha+1}\theta\|\|\rho\|_{\dot H^{s-1,6}}(\|\nabla\Delta u\|_{L^3}+\|\nabla\rho\|_{L^{\infty}}\|\Delta u\|_{L^3})\\
\leq & \delta_0\kappa\|\partial^{\alpha+1}\theta\|^2 +\frac{C\hbar^4}{\delta_0\kappa}\|\Delta u\|_{\dot H^s}^2 +\frac{CE^2\hbar^2}{\delta_0\kappa}\|\rho\|_{\dot H^s}^2 +\frac{CE^2\hbar^2}{\delta_0\kappa}\|\Delta u\|_{\dot H^{s-1}}^2.
\end{split}
\end{equation*}
Similar to $R_4$, we have for the term $R_5$ that
\begin{equation*}
\begin{split}
R_5 =& \frac{2}{3}\int\partial^{\alpha}\left(\frac{\mu|\nabla u+(\nabla u)^T|^2+2\lambda(\text{div} u)^2}{(1+\rho)}\right)\partial^{\alpha}\theta\\
\leq & \delta_0\kappa\|\partial^{\alpha+1}\theta\|^2 +\frac{CE^2}{\delta_0\kappa}\|(\nabla\rho,\nabla u)\|_{\dot H^{s-1}}^2.
\end{split}
\end{equation*}
Putting these estimates together, we obtain
\begin{equation*}
\begin{split}
\frac12\frac{d}{dt}\|\partial^{\alpha}\theta\|^2 &+\frac{2\kappa}{3}\int\frac{|\partial^{\alpha}\nabla\theta|^2}{(1+\rho)} \leq \left(C\kappa E+\delta_0\kappa\right) \|\partial^{\alpha}\nabla\theta\|^2 +\delta_0(\mu+\lambda)\|\partial^{\alpha}\nabla\cdot u\|_{L^2}^2\\
&+\frac{C}{\delta_0}\|\partial^{\alpha}\theta\|_{L^2}^2 +\frac{C\hbar^4}{\delta_0\kappa}\|\partial^{\alpha}\Delta u\|^2 +\frac{CE^2}{\delta_0\kappa}\|(\nabla\rho,\nabla u)\|_{\dot H^{s-1}}^2 +\frac{CE^2\hbar^2}{\delta_0\kappa}\|\Delta u\|_{\dot H^{s-1}}^2.
\end{split}
\end{equation*}
Integrating this inequality in time over $[0,t]$ and noting $1/2<1+\rho<3/2$, we know that there exists some constants $\delta_0<1$ and $\varepsilon_0<1$ sufficiently small, such that \eqref{equ27} holds for all $\delta\leq\delta_0$ and $E\leq \varepsilon\leq \varepsilon_0$.
\end{proof}

\begin{lemma}\label{lem9}
Under the assumptions in Proposition \ref{prop3}, there exists some constant $\varepsilon_0<1$ sufficiently small and $\hbar_0<1$, such that
\begin{equation}\label{equ28}
\begin{split}
&\|\partial^{\alpha}(\rho,u,\hbar\nabla\rho)(t)\|^2 +\nu_0\int_0^t\|\partial^{\alpha}\nabla u(\tau)\|^2d\tau\\
\leq & C\|\partial^{\alpha}(\rho,u,\hbar\nabla\rho)(0)\|^2 +CE\int_0^t\|\nabla(\rho,u)(\tau)\|_{\dot H^s}^2d\tau +C\int_0^t\|\nabla(\rho,u,\theta)(\tau)\|_{\dot H^{s-1}}^2d\tau.
\end{split}
\end{equation}
for all $E\leq \varepsilon\leq \varepsilon_0$ and $\hbar<\hbar_0$, where $C$ is independent of $t$.
\end{lemma}
\begin{proof}
Applying $\partial^{\alpha}$ to \eqref{QF2} and then taking inner product of the resultant with $\partial^{\alpha}u$, we obtain
\begin{equation}\label{equ13}
\begin{split}
\sum L_i=&\frac12\frac{d}{dt}\|\partial^{\alpha}u\|^2-\mu\int\partial^{\alpha}(\frac{\Delta u}{\rho+1})\partial^{\alpha}u -(\mu+\lambda)\int\partial^{\alpha}(\frac{\nabla\text{div}u}{\rho+1})\partial^{\alpha}u\\
=& -\int\partial^{\alpha}(u\cdot\nabla u)\partial^{\alpha}u -\int\partial^{\alpha}\nabla\theta\partial^{\alpha}u -\int\partial^{\alpha}(\frac{\theta+1}{\rho+1}\nabla\rho)\partial^{\alpha}u\\
&+\frac{\hbar^2}{12}\int\partial^{\alpha}(\frac{\Delta\nabla\rho}{\rho+1})\partial^{\alpha}u -\frac{\hbar^2}{3}\int\partial^{\alpha}(\frac{\text{div}\{\cdots\}}{\rho+1}) \partial^{\alpha}u =\sum R_i
\end{split}
\end{equation}
Now, for the term $L_2$, we have by integration by parts
\begin{equation*}
\begin{split}
L_2= & -\mu\int\frac{\partial^{\alpha}\Delta u}{1+\rho}\partial^{\alpha}u -\mu\int[\partial^{\alpha},\frac{1}{1+\rho}]\Delta u\partial^{\alpha}u\\
= & \mu\int\frac{|\partial^{\alpha}\nabla u|^2}{1+\rho} -\mu\int\frac{\partial^{\alpha}\nabla u}{(1+\rho)^2}\nabla\rho\partial^{\alpha}u -\mu\int[\partial^{\alpha},\frac{1}{1+\rho}]\Delta u\partial^{\alpha}u.
\end{split}
\end{equation*}
Invoking Lemma \ref{lem0} and \ref{Le-inequ}, we obtain
\begin{equation*}
\begin{split}
\|[\partial^{\alpha},\frac{1}{1+\rho}]\Delta u\|_{L^{6/5}}\leq & C\|\nabla\rho\|_{L^3}\|\partial^{\alpha-1}\Delta u\|_{L^2} +C\|\Delta u\|_{L^3}\|\partial^{\alpha}(\frac{1}{1+\rho})\|_{L^2}\\
\leq & CE\|\partial^{\alpha-1}\Delta u\|_{L^2} +CE\|\nabla\rho\|_{\dot H^{s-1}}.
\end{split}
\end{equation*}
Hence
\begin{equation*}
\begin{split}
L_2\geq & \frac{2\mu}{3}\|\partial^{\alpha}\nabla u\|^2
-C\mu\|\partial^{\alpha}\nabla u\|_{L^2}\|\nabla\rho\|_{L^3}\|\partial^{\alpha}u\|_{L^6}
-C\|[\partial^{\alpha},\frac{1}{1+\rho}]\Delta u\|_{L^{6/5}}\|\partial^{\alpha}u\|_{L^6}\\
\geq &\frac{2\mu}{3}\|\partial^{\alpha}\nabla u\|^2 -CE(1+\mu)\|\partial^{\alpha}\nabla u\|^2 -CE\|\nabla\rho\|_{\dot H^{s-1}}^2\\
\geq & \frac{\mu}{2}\|\partial^{\alpha}\nabla u\|^2-CE\|\nabla\rho\|_{\dot H^{s-1}}^2
\end{split}
\end{equation*}
by taking $E\leq \mu/6C(1+\mu)$. Similarly, we have for $L_3$ that
\begin{equation*}
\begin{split}
L_3 \geq & \frac{\mu+\lambda}{2}\|\partial^{\alpha}\nabla\cdot u\|^2-CE\|\nabla\rho\|_{\dot H^{s-1}}^2
\end{split}
\end{equation*}
For the RHS term $R_1$, we have by integration by parts that
\begin{equation*}
\begin{split}
R_1 
=& \frac12\int \nabla\cdot u|\partial^{\alpha}u|^2 -\int[\partial^{\alpha},u]\nabla u\partial^{\alpha}u\\
\leq & C\|\nabla u\|_{L^{\infty}}\|\partial^{\alpha}u\|^2 \leq CE\|\nabla u\|_{\dot H^{s-1}}^2.
\end{split}
\end{equation*}
For the term $R_2$, we have for any $\delta_0>0$ that
\begin{equation*}
\begin{split}
R_2\leq \delta_0\mu\|\nabla\partial^{\alpha}u\|^2 +\frac{C}{\delta_0\mu}\|\nabla\theta\|_{\dot H^{s-1}}^2.
\end{split}
\end{equation*}
The term $R_3$ will be treated with much more effort, from which some good terms will appear. By integration by parts,
\begin{equation*}
\begin{split}
R_3= & -\int\frac{\theta+1}{\rho+1}\nabla\partial^{\alpha}\rho\partial^{\alpha}u -\int\partial^{\alpha}u[\partial^{\alpha},\frac{\theta+1}{\rho+1}]\nabla\rho\\
= & \int\frac{\theta+1}{\rho+1}\partial^{\alpha}\rho\partial^{\alpha}\text{div}u +\int\nabla(\frac{\theta+1}{\rho+1})\partial^{\alpha}\rho\partial^{\alpha}u -\int\partial^{\alpha}u[\partial^{\alpha},\frac{\theta+1}{\rho+1}]\nabla\rho =\sum R_{3i}.
\end{split}
\end{equation*}
It is easy to show that
\begin{equation*}
\begin{split}
R_{32} \leq C\|\nabla(\theta,\rho)\|_{L^{\infty}}\|\partial^{\alpha}(\rho,u)\|^2 \leq CE\|\nabla(\rho,u)\|_{\dot H^{s-1}}^2,
\end{split}
\end{equation*}
and
\begin{equation*}
\begin{split}
R_{33} \leq & C\|\partial^{\alpha}u\|\left(\|\nabla(\frac{1+\theta}{1+\rho})\|_{L^{\infty}} \|\nabla\rho\|_{\dot H^{s-1}} +\|\nabla\rho\|_{L^{\infty}}\|\frac{1+\theta}{1+\rho}\|_{\dot H^s}\right)\\
\leq & CE\|\nabla(\rho,u,\theta)\|_{\dot H^{s-1}}^2,
\end{split}
\end{equation*}
thanks to Lemma \ref{lem0}. Differentiating the continuity equation \eqref{QF1} with $\partial^{\alpha}$ yields
\begin{equation*}
\begin{split}
(1+\rho)\partial^{\alpha}\text{div}u=-\partial_t\partial^{\alpha}\rho -\partial^{\alpha}(u\cdot\nabla\rho) -[\partial^{\alpha},1+\rho]\text{div}u,
\end{split}
\end{equation*}
which implies that
\begin{equation*}
\begin{split}
R_{31}= & -\int\frac{\theta+1}{(1+\rho)^2}\partial^{\alpha}\rho \{\partial_t\partial^{\alpha}\rho +\partial^{\alpha}(u\cdot\nabla\rho) +[\partial^{\alpha},1+\rho]\text{div}u\}=\sum R_{31i}.
\end{split}
\end{equation*}
It is immediately from \eqref{equ11} and \eqref{equ29} with $p=3/2$ that
\begin{equation*}
\begin{split}
R_{311}= & -\frac12\frac{d}{dt}\int\frac{\theta+1}{(1+\rho)^2}|\partial^{\alpha}\rho|^2 +\frac12\int\partial_t(\frac{\theta+1}{(1+\rho)^2})|\partial^{\alpha}\rho|^2\\
\leq & -\frac12\frac{d}{dt}\int\frac{\theta+1}{(1+\rho)^2}|\partial^{\alpha}\rho|^2 +C\|\partial_t\theta,\partial_t\rho\|_{L^{3/2}} \|\partial^{\alpha}\rho\|_{L^6}^2\\
\leq & -\frac12\frac{d}{dt}\int\frac{\theta+1}{(1+\rho)^2}|\partial^{\alpha}\rho|^2 +CE\|\partial^{\alpha}\nabla\rho\|^2.
\end{split}
\end{equation*}
For the term $R_{312}$, we have by integration by parts that
\begin{equation*}
\begin{split}
R_{312} =&  \frac12\int\text{div}(\frac{(\theta+1)u}{(1+\rho)^2})|\partial^{\alpha}\rho|^2 -\int\frac{(\theta+1)}{(1+\rho)^2}\partial^{\alpha}\rho[\partial^{\alpha},u] \nabla\rho \\
\leq & C\|\nabla(\rho,u,\theta)\|_{L^{\infty}}\|\partial^{\alpha}\rho\|^2 +\|\partial^{\alpha}\rho\|\left(\|\nabla u\|_{L^{\infty}}\|\nabla\rho\|_{\dot H^{s-1}} \|\nabla\rho\|_{L^{\infty}}\|u\|_{\dot H^{s}}\right)\\
\leq & CE\|\nabla(\rho,u)\|_{\dot H^{s-1}}^2.
\end{split}
\end{equation*}
The same estimate hold for $R_{313}$. Combining all the estimates for $R_3$, we obtain
\begin{equation*}
\begin{split}
R_{3} \leq & -\frac12\frac{d}{dt}\int\frac{\theta+1}{(1+\rho)^2}|\partial^{\alpha}\rho|^2 +CE\|\partial^{\alpha}\nabla\rho\|^2 +CE\|\nabla(\rho,u,\theta)\|_{\dot H^{s-1}}^2.
\end{split}
\end{equation*}
Now, we consider the estimate of $R_4$. By integration by parts, we obtain
\begin{equation*}
\begin{split}
R_4=&\frac{\hbar^2}{12}\int\partial^{\alpha}\partial_i\partial_j (\frac{\partial_i\rho}{\rho+1}) \partial^{\alpha}u^j -\frac{\hbar^2}{12}\int\partial^{\alpha}\left([\partial_i\partial_j, \frac{1}{\rho+1}]\partial_i\rho\right)\partial^{\alpha}u^j\\
= &\frac{\hbar^2}{12}\int\partial^{\alpha}(\frac{\nabla\rho}{\rho+1}) \partial^{\alpha}\nabla\text{div}u +\frac{\hbar^2}{12}\int\partial^{\alpha-1}\left([\nabla^2, \frac{1}{\rho+1}]\nabla\rho\right)\partial^{\alpha+1}u\\
=& \frac{\hbar^2}{12}\int\frac{\partial^{\alpha}\nabla\rho}{\rho+1} \partial^{\alpha}\nabla\text{div}u +\frac{\hbar^2}{12}\int[\partial^{\alpha},\frac{1}{\rho+1}]\nabla\rho \partial^{\alpha}\nabla\text{div}u\\
& +\frac{\hbar^2}{12}\int\partial^{\alpha-1}\left([\nabla^2, \frac{1}{\rho+1}]\nabla\rho\right)\partial^{\alpha+1}u =:\sum R_{4i}.
\end{split}
\end{equation*}
Using the continuity equation \eqref{QF1} and similar to the term $R_{31}$, it can be shown that
\begin{equation*}
\begin{split}
R_{41}= & -\frac{\hbar^2}{12}\int\frac{\partial^{\alpha}\nabla\rho}{(1+\rho)^2} \left\{\partial_t\partial^{\alpha}\nabla\rho +\partial^{\alpha}\nabla(u\cdot\nabla\rho) +[\partial^{\alpha}\nabla,1+\rho]\text{div}u\right\}=\sum R_{41i}.
\end{split}
\end{equation*}
For $R_{411}$, we obtain
\begin{equation*}
\begin{split}
R_{411}= & -\frac{\hbar^2}{24}\frac{d}{dt}\int\frac{|\partial^{\alpha}\nabla\rho|^2}{(1+\rho)^2} -\frac{\hbar^2}{12}\frac{d}{dt}\int\frac{\partial_t\rho|\partial^{\alpha}\nabla\rho|^2} {(1+\rho)^3}\\
\leq &-\frac{\hbar^2}{24}\frac{d}{dt}\int\frac{|\partial^{\alpha}\nabla\rho|^2}{(1+\rho)^2} +C\hbar^2E\|\partial^{\alpha}\nabla\rho\|^2.
\end{split}
\end{equation*}
By integration by parts,
\begin{equation*}
\begin{split}
R_{412}= & -\frac{\hbar^2}{12}\int\frac{\partial^{\alpha}\nabla\rho}{(1+\rho)^2} u\cdot\partial^{\alpha}\nabla^2\rho -\frac{\hbar^2}{12}\int\frac{\partial^{\alpha}\nabla\rho}{(1+\rho)^2} [\partial^{\alpha}\nabla,u]\nabla\rho\\
= & \frac{\hbar^2}{24}\int\text{div}\left(\frac{u}{(1+\rho)^2}\right) |\partial^{\alpha}\nabla\rho|^2 -\frac{\hbar^2}{12}\int\frac{\partial^{\alpha}\nabla\rho}{(1+\rho)^2} [\partial^{\alpha}\nabla,u]\nabla\rho,
\end{split}
\end{equation*}
and hence by Lemma \ref{Le-inequ} and \eqref{equ10}
\begin{equation*}
\begin{split}
R_{412}\leq  & C\hbar^2\|\nabla(\rho,u)\|_{L^{\infty}}\|\partial^{\alpha}\nabla\rho\|^2 \\
&+C\hbar^2\|\partial^{\alpha}\nabla\rho\|(\|\partial^{\alpha}\nabla\rho\|\|\nabla u\|_{L^{\infty}}+\|\nabla\rho\|_{L^{\infty}}\|\partial^{\alpha}\nabla u\|)\\
\leq & C\hbar^2E\|\partial^{\alpha}\nabla(\rho,u)\|^2.
\end{split}
\end{equation*}
Similarly, by Lemma \ref{Le-inequ} and \eqref{equ10},
\begin{equation*}
\begin{split}
R_{413} \leq & C\hbar^2E\|\partial^{\alpha}\nabla(\rho,u)\|^2.
\end{split}
\end{equation*}
For the term $R_{42}$, we have
\begin{equation*}
\begin{split}
R_{42}= & -\frac{\hbar^2}{12}\int\nabla([\partial^{\alpha},\frac{1}{\rho+1}]\nabla\rho) \partial^{\alpha}\text{div}u\\
=&\frac{\hbar^2}{12}\int([\partial^{\alpha},\frac{\nabla\rho}{(\rho+1)^2}]\nabla\rho) \partial^{\alpha}\text{div}u -\frac{\hbar^2}{12}\int[\partial^{\alpha},\frac{1}{\rho+1}]\nabla^2\rho \partial^{\alpha}\text{div}u
\end{split}
\end{equation*}
But by the commutator estimates, we have
\begin{equation*}
\begin{split}
\|[\partial^{\alpha},\frac{\nabla\rho}{(\rho+1)^2}]\nabla\rho\|_{L^2}\leq & \|\nabla\rho\|_{\dot H^{s-1,6}}\|\frac{\nabla\rho}{(\rho+1)^2}\|_{L^3} +\|\nabla\rho\|_{L^{\infty}}\|\frac{\nabla\rho}{(\rho+1)^2}\|_{\dot H^s}\\
\leq & CE\|\nabla\rho\|_{\dot H^s},
\end{split}
\end{equation*}
and
\begin{equation*}
\begin{split}
\|[\partial^{\alpha},\frac{1}{\rho+1}]\nabla^2\rho \|_{L^2}\leq & \|\nabla^2\rho\|_{\dot H^{s-1}}\|\nabla(\frac{1}{\rho+1})\|_{L^\infty} +\|\nabla^2\rho\|_{L^{3}}\|\frac{1}{\rho+1}\|_{\dot H^{s,6}}\\
\leq & CE\|\nabla\rho\|_{\dot H^s},
\end{split}
\end{equation*}
hence
\begin{equation*}
\begin{split}
R_{42}\leq C\hbar^2E\|\nabla(\rho,u)\|_{\dot H^s}^2.
\end{split}
\end{equation*}
Similarly, for $R_{43}$, we obtain
\begin{equation*}
\begin{split}
R_{43}\leq C\hbar^2E\|\nabla(\rho,u)\|_{\dot H^s}^2.
\end{split}
\end{equation*}
Putting all the estimates for $R_4$ together, we obtain
\begin{equation*}
\begin{split}
R_{4} \leq &-\frac{\hbar^2}{24}\frac{d}{dt}\int\frac{|\partial^{\alpha}\nabla\rho|^2}{(1+\rho)^2} +C\hbar^2E\|\nabla(\rho,u)\|_{\dot H^s}^2.
\end{split}
\end{equation*}
Finally, for $R_5$, it is easy to show
\begin{equation*}
\begin{split}
R_5= & \frac{\hbar^2}{12}\int\partial^{\alpha-1}(\frac{\text{div} \{\nabla\rho\otimes\nabla\rho/(1+\rho)\}}{\rho+1}) \partial^{\alpha+1}u\\
\leq & C\hbar^2E\|\partial^{\alpha}\nabla(\rho,u)\|^2 +C\hbar^2E\|\nabla\rho\|_{\dot H^{s-1}}^2.
\end{split}
\end{equation*}
Now, putting all these estimates for \eqref{equ13} together, and taking $\delta_0=1/4$, we obtain,
\begin{equation*}
\begin{split}
\frac12\frac{d}{dt}\|\partial^{\alpha}u\|^2 &+\frac12\frac{d}{dt}\int\frac{\theta+1}{(1+\rho)^2}|\partial^{\alpha}\rho|^2 +\frac{\hbar^2}{24}\frac{d}{dt}\int\frac{|\partial^{\alpha}\nabla\rho|^2}{(1+\rho)^2}\\
&+\frac{\mu}{4}\|\partial^{\alpha}\nabla u\|^2 +\frac{\mu+\lambda}{2}\|\partial^{\alpha}\nabla\cdot u\|^2\\
\leq & C\|\partial^{\alpha}\theta\|^2 +CE\|\partial^{\alpha}\nabla\rho\|^2 +CE\|\nabla(\rho,u,\theta)\|_{\dot H^{s-1}}^2 +C\hbar^2E\|\partial^{\alpha}\nabla u\|^2.
\end{split}
\end{equation*}
Integrating in time over $[0,t]$ completes the proof, thanks to $\hbar<1$.
\end{proof}

\begin{lemma}\label{lem10}
Under the assumptions in Proposition \ref{prop3}, there exists some constant $\varepsilon_0<1$ sufficiently small and $\hbar_0<1$, such that
\begin{equation}\label{equ30}
\begin{split}
& \hbar^2\|\nabla\partial^{\alpha}(u,\rho,\hbar\nabla\rho)(t)\|^2 +\nu_0\hbar^2\int_0^t\|\Delta\partial^{\alpha}u(\tau)\|^2d\tau\\
\leq & C\hbar^2\|\nabla\partial^{\alpha}(u,\rho,\hbar\nabla\rho)(0)\|^2 +C\hbar^4E\int_0^t\|\Delta\partial^{\alpha}u(\tau)\|^2d\tau +C\hbar^2E\int_0^t\|\Delta\partial^{\alpha}\rho(\tau)\|^2d\tau \\
& +C\hbar^2\int_0^t\|\partial^{\alpha}\nabla\theta(\tau)\|^2d\tau +C\hbar^2E\int_0^t\|\nabla(\rho,u,\theta)(\tau)\|_{\dot H^{s}}^2d\tau.
\end{split}
\end{equation}
for all $E\leq \varepsilon\leq \varepsilon_0$ and $\hbar<\hbar_0$, where $C$ is independent of $t$.
\end{lemma}
\begin{proof}
Applying $\partial^{\alpha}$ to \eqref{QF2} and then taking inner product of the resultant with $-\hbar^2\Delta\partial^{\alpha}u$, we obtain
\begin{equation*}
\begin{split}
\sum L_i=&\frac{\hbar^2}2\frac{d}{dt}\|\nabla\partial^{\alpha}u\|^2 +\mu\hbar^2\int\partial^{\alpha}(\frac{\Delta u}{\rho+1})\Delta\partial^{\alpha}u +(\mu+\lambda)\hbar^2\int\partial^{\alpha}(\frac{\nabla\text{div}u}{\rho+1}) \Delta\partial^{\alpha}u\\
=& \hbar^2\int\partial^{\alpha}(u\cdot\nabla u)\Delta\partial^{\alpha}u +\hbar^2\int\partial^{\alpha}\nabla\theta\Delta\partial^{\alpha}u +\hbar^2\int\partial^{\alpha}(\frac{\theta+1}{\rho+1}\nabla\rho) \Delta\partial^{\alpha}u\\
&-\frac{\hbar^4}{12}\int\partial^{\alpha}(\frac{\Delta\nabla\rho}{\rho+1}) \Delta\partial^{\alpha}u +\frac{\hbar^4}{3}\int\partial^{\alpha}(\frac{\text{div}\{\cdots\}}{\rho+1}) \Delta\partial^{\alpha}u =\sum R_i
\end{split}
\end{equation*}
Now, for the term $L_2$, we have by integration by parts
\begin{equation*}
\begin{split}
L_2= & \mu\hbar^2\int\frac{|\Delta\partial^{\alpha}u|^2}{1+\rho} +\mu\hbar^2\int[\partial^{\alpha},\frac{1}{1+\rho}]\Delta u\Delta\partial^{\alpha}u.
\end{split}
\end{equation*}
Since
\begin{equation*}
\begin{split}
\|[\partial^{\alpha},\frac{1}{1+\rho}]\Delta u\|_{L^{2}}\leq & C\|\nabla\rho\|_{L^{\infty}}\|\partial^{\alpha-1}\Delta u\|_{L^2} +C\|\Delta u\|_{L^{3}}\|\partial^{\alpha}(\frac{1}{1+\rho})\|_{L^6}\\
\leq & CE\|\Delta u\|_{\dot H^{s-1}} +C\|\Delta u\|_{H^1}\|\rho\|_{\dot H^{s,6}}\\
\leq & CE\|(\rho,u)\|_{\dot H^{s+1}},
\end{split}
\end{equation*}
thanks to Lemma \ref{lem0} and \eqref{Le-inequ}, we have
\begin{equation*}
\begin{split}
L_2\geq \frac{\mu\hbar^2}{2}\|\Delta\partial^{\alpha}u\|^2 -C\hbar^2E^2\|(\rho,u)\|_{\dot H^{s+1}}^2.
\end{split}
\end{equation*}
Similarly, for the $L_3$, we have
\begin{equation*}
\begin{split}
L_3\geq \frac{(\mu+\lambda)\hbar^2}{2}\|\partial^{\alpha}\nabla\text{div}u\|^2 -C\hbar^2E^2\|(\rho,u)\|_{\dot H^{s+1}}^2.
\end{split}
\end{equation*}
For the RHS term $R_1$, we have by integration by parts twice that
\begin{equation*}
\begin{split}
R_1 = &-\hbar^2\int\partial_j\partial^{\alpha}(u^i\cdot\partial_i u)\partial_j\partial^{\alpha}u\\
= &-\hbar^2\int u^i\partial_j\partial^{\alpha}\partial_iu\cdot\partial_j\partial^{\alpha}u -\hbar^2\int[\partial_j\partial^{\alpha},u^i]\partial_iu\partial_j\partial^{\alpha}u\\
= & \frac{\hbar^2}{2}\int \partial_iu^i|\partial_j\partial^{\alpha}u|^2 -\hbar^2\int[\partial_j\partial^{\alpha},u^i]\partial_iu\partial_j\partial^{\alpha}u\\
\leq & C\hbar^2E\|u\|_{\dot H^{s+1}}^2.
\end{split}
\end{equation*}
For the term $R_2$, we have
\begin{equation*}
\begin{split}
R_2=\hbar^2\int\partial^{\alpha}\nabla\theta\Delta\partial^{\alpha}u \leq \delta_0\mu\hbar^2\|\Delta\partial^{\alpha}u\|^2 +\frac{C\hbar^2}{\delta_0\mu}\|\partial^{\alpha}\nabla\theta\|^2.
\end{split}
\end{equation*}
By integration by parts,
\begin{equation*}
\begin{split}
R_3= & -\hbar^2\int\nabla\partial^{\alpha}(\frac{\theta+1}{\rho+1}\nabla\rho) \nabla\partial^{\alpha}u\\
=& -\hbar^2\int\frac{\theta+1}{\rho+1}\nabla\partial^{\alpha}\nabla\rho\nabla \partial^{\alpha}u -\hbar^2\int\nabla\partial^{\alpha}u[\nabla\partial^{\alpha},\frac{\theta+1}{\rho+1}]\nabla\rho\\
= & \hbar^2\int\frac{\theta+1}{\rho+1}\nabla\partial^{\alpha}\rho\nabla\partial^{\alpha}\text{div}u +\hbar^2\int\nabla(\frac{\theta+1}{\rho+1})\nabla\partial^{\alpha}\rho\nabla\partial^{\alpha}u\\
& -\hbar^2\int\nabla\partial^{\alpha}u[\nabla\partial^{\alpha},\frac{\theta+1}{\rho+1}]\nabla\rho =\sum R_{3i}.
\end{split}
\end{equation*}
It is easy to show that
\begin{equation*}
\begin{split}
R_{32} \leq \hbar^2\|\nabla(\theta,\rho)\|_{L^{\infty}}\|\partial^{\alpha}\nabla(\rho,u)\|^2 \leq C\hbar^2E\|\nabla\partial^{\alpha}(\rho,u)\|^2
\end{split}
\end{equation*}
and by Lemma \ref{lem0} and \ref{Le-inequ}
\begin{equation*}
\begin{split}
R_{33} \leq & C\hbar^2\|\nabla\partial^{\alpha}u\| \|[\nabla\partial^{\alpha},\frac{\theta+1}{\rho+1}]\nabla\rho\|\\
\leq & C\hbar^2\|\nabla\partial^{\alpha}u\| (\|\nabla\rho\|_{\dot H^s}\|\nabla(\rho,\theta)\|_{L^{\infty}} +\|\nabla\rho\|_{L^{\infty}}\|(\rho,\theta)\|_{\dot H^{s+1}})\\
\leq & C\hbar^2E\|\nabla(\rho,u,\theta)\|_{\dot H^{s}}^2.
\end{split}
\end{equation*}
Differentiating the continuity equation \eqref{QF1} with $\partial^{\alpha}$ and then inserting the resultant to $R_{31}$, we obtain
\begin{equation*}
\begin{split}
R_{31}= & -\hbar^2\int\frac{\theta+1}{(1+\rho)^2}\nabla\partial^{\alpha}\rho \{\partial_t\nabla\partial^{\alpha}\rho +\nabla\partial^{\alpha}(u\cdot\nabla\rho) +[\nabla\partial^{\alpha},1+\rho]\text{div}u\}=\sum R_{31i}.
\end{split}
\end{equation*}
It is immediately that
\begin{equation*}
\begin{split}
R_{311}= & -\frac{\hbar^2}2\frac{d}{dt}\int\frac{\theta+1}{(1+\rho)^2}|\nabla\partial^{\alpha}\rho|^2 +\frac{\hbar^2}2\int\partial_t(\frac{\theta+1}{(1+\rho)^2})|\nabla\partial^{\alpha}\rho|^2
\\ \leq & -\frac{\hbar^2}2\frac{d}{dt}\int\frac{\theta+1}{(1+\rho)^2}|\nabla\partial^{\alpha}\rho|^2 +C\hbar^2E\|\partial^{\alpha}\Delta\rho\|^2,
\end{split}
\end{equation*}
thanks to \eqref{equ11} and \eqref{equ29} again. For the term $R_{312}$, we have by integration by parts that
\begin{equation*}
\begin{split}
R_{312} =&  \frac{\hbar^2}2\int\text{div}(\frac{(\theta+1)u}{(1+\rho)^2})|\nabla\partial^{\alpha}\rho|^2 -\hbar^2\int\frac{(\theta+1)}{(1+\rho)^2}\nabla\partial^{\alpha}\rho[\nabla\partial^{\alpha},u] \nabla\rho \\
\leq & C\hbar^2E\|\nabla\partial^{\alpha}(\rho,u)\|^2.
\end{split}
\end{equation*}
The same estimate hold for $R_{313}$. Now, we consider the estimate of $R_4$. By integration by parts, we obtain
\begin{equation*}
\begin{split}
R_4=& -\frac{\hbar^4}{12}\int\frac{\Delta\partial^{\alpha}\nabla\rho}{\rho+1} \Delta\partial^{\alpha}u -\frac{\hbar^4}{12}\int[\partial^{\alpha},\frac{1}{\rho+1}]\Delta\nabla\rho \Delta\partial^{\alpha}u\\
= & \frac{\hbar^4}{12}\int\frac{\Delta\partial^{\alpha}\rho}{\rho+1} \Delta\partial^{\alpha}\text{div}u -\frac{\hbar^4}{12}\int\frac{\Delta\partial^{\alpha}\rho\Delta\partial^{\alpha} u}{(\rho+1)^2}\nabla\rho -\frac{\hbar^4}{12}\int[\partial^{\alpha},\frac{1}{\rho+1}]\Delta\nabla\rho \Delta\partial^{\alpha}u\\
= & \sum R_{4i}.
\end{split}
\end{equation*}
For the last two terms, it can be shown that
\begin{equation*}
\begin{split}
R_{42}+R_{43}\leq & C\hbar^4\|\nabla\rho\|_{\infty}\|\Delta\partial^{\alpha}(\rho,u)\|^2 +C\hbar^4\|\Delta\partial^{\alpha}u\|\|\Delta\nabla\rho\|_{L^3} \|\partial^{\alpha}(1/(1+\rho))\|_{L^6}\\
\leq & C\hbar^4E\|\Delta\partial^{\alpha}(\rho,u)\|^2  +C\hbar^4E\|\nabla\rho\|_{\dot H^{s}}^2.
\end{split}
\end{equation*}
Using the continuity equation \eqref{QF1} and similar to the term $R_{31}$, it can be shown that
\begin{equation*}
\begin{split}
R_{41}= & -\frac{\hbar^4}{12}\int\frac{\Delta\partial^{\alpha}\rho}{(1+\rho)^2} \{\partial_t\partial^{\alpha}\Delta\rho +\partial^{\alpha}\Delta(u\cdot\nabla\rho) +[\partial^{\alpha}\Delta,1+\rho]\text{div}u\}\\
\leq & -\frac{\hbar^4}{24}\frac{d}{dt}\int\frac{|\partial^{\alpha}\Delta\rho|^2}{(1+\rho)^2} +C\hbar^4E\|\partial^{\alpha}\Delta(\rho,u)\|^2.
\end{split}
\end{equation*}
Finally, for $R_5$, we have
\begin{equation*}
\begin{split}
R_5= & \frac{\hbar^4}{12}\int\partial^{\alpha}(\frac{\text{div} \{\nabla\rho\otimes\nabla\rho/(1+\rho)\}}{\rho+1})\Delta\partial^{\alpha}u\\
\leq & C\hbar^4E\|\partial^{\alpha}\Delta(\rho,u)\|^2 +C\hbar^4E\|\nabla\rho\|_{\dot H^{s}}^2.
\end{split}
\end{equation*}
Now, putting all these estimates together, and taking $\delta_0=1/4$, we obtain,
\begin{equation*}
\begin{split}
\frac{\hbar^2}2\frac{d}{dt}&\|\nabla\partial^{\alpha}u\|^2 +\frac{\mu\hbar^2}{2}\|\Delta\partial^{\alpha}u\|^2 +\frac{(\mu+\lambda)\hbar^2}{2}\|\partial^{\alpha}\nabla\text{div}u\|^2\\
&+\frac{\hbar^2}2\frac{d}{dt}\int\frac{\theta+1}{(1+\rho)^2}|\nabla\partial^{\alpha}\rho|^2 +\frac{\hbar^4}{24}\frac{d}{dt}\int\frac{|\partial^{\alpha}\Delta\rho|^2}{(1+\rho)^2}\\
\leq & C\hbar^4E\|\Delta\partial^{\alpha}u\|^2 +C\hbar^2E\|\Delta\partial^{\alpha}\rho\|^2 +C\hbar^2\|\partial^{\alpha}\nabla\theta\|^2 +C\hbar^2E\|\nabla(\rho,u,\theta)\|_{\dot H^{s}}^2.
\end{split}
\end{equation*}
Integrating in time over $[0,t]$ completes the proof.
\end{proof}

\begin{lemma}\label{lem11}
Under the assumptions in Proposition \ref{prop3}, there exists some constant $\varepsilon_0<1$ sufficiently small and $\hbar_0<1$, such that
\begin{equation*}
\begin{split}
&(2\mu+\lambda)\beta\|\partial^{\gamma}\Delta\rho(t)\|^2 +\beta\int_0^t\|\partial^{\gamma}\Delta(\rho,\hbar\nabla\rho)(\tau)\|^2d\tau\\
\leq &C\beta\|\partial^{\gamma}\Delta\rho_0\|^2 +\beta\|\partial^{\gamma}\nabla u_0\|^2 +C\beta\|\partial^{\gamma}\text{div}u(t)\|^2\\
&+C\beta\int_0^t\|\partial^{\gamma}\Delta\theta\|^2d\tau +C\beta\int_0^t\|\partial^{\alpha}\nabla u\|_{L^2}^2d\tau
+C\beta\int_0^t\|\nabla(\rho,u,\theta)\|^2_{\dot H^{s-1}}d\tau,
\end{split}
\end{equation*}
for any $\beta>0$, $E\leq \varepsilon\leq \varepsilon_0$ and $\hbar\leq\hbar_0$, where $C$ depends only on $\mu,\lambda$ and some Sobolev constants. In particular, $C$ does not depend on $\hbar>0$ or $t>0$.
\end{lemma}
\begin{proof}
Let $\gamma$ be a multi-index such that $|\gamma|=|\alpha|-1$ and $\gamma\leq \alpha$. We apply $\partial^{\gamma}$ to \eqref{QF2} and then take the inner product of the resultant with $-\beta\partial^{\gamma}\nabla\Delta\rho$ to obtain
\begin{equation}\label{equ41}
\begin{split}
\sum_{i=1}^8L_i=&-\beta\int \partial^{\gamma}u_t\cdot\partial^{\gamma}\nabla\Delta\rho +\beta\mu\int\partial^{\gamma}\left(\frac{\Delta u}{1+\rho}\right)\partial^{\gamma}\nabla\Delta\rho \\
&+\beta(\mu+\lambda)\int\partial^{\gamma}\left(\frac{\nabla\text{div}u} {1+\rho}\right)\partial^{\gamma}\nabla\Delta\rho -\beta\int\partial^{\gamma}(u\cdot\nabla u)\partial^{\gamma}\nabla\Delta\rho \\
& -\beta\int\partial^{\gamma}\nabla\theta\partial^{\gamma}\nabla\Delta\rho -\beta\int\partial^{\gamma}\left(\frac{\theta+1}{1+\rho}\nabla\rho\right) \partial^{\gamma}\nabla\Delta\rho \\
& +\frac{\beta\hbar^2}{12}\int \partial^{\gamma}\left(\frac{\nabla\Delta\rho}{1+\rho}\right)\partial^{\gamma}\nabla\Delta\rho -\frac{\beta\hbar^2}{3}\int\partial^{\gamma}\left(\frac{\text{div}\{\cdots\}} {1+\rho}\right)\partial^{\gamma}\nabla\Delta\rho=0.
\end{split}
\end{equation}
We first note that by H\"older inequality
\begin{equation*}
\begin{split}
L_7= & \frac{\beta\hbar^2}{12}\int\frac{|\partial^{\gamma}\nabla\Delta\rho|^2}{1+\rho} +\frac{\beta\hbar^2}{12}\int[\partial^{\gamma},\frac{1}{1+\rho}]\nabla\Delta\rho \partial^{\gamma}\nabla\Delta\rho\\
\geq & \frac{\beta\hbar^2}{18}\|\partial^{\gamma}\nabla\Delta\rho\|^2 -C\beta\hbar^2\|\partial^{\gamma}\nabla\Delta\rho\| (\|\nabla\Delta\rho\|_{\dot H^{s-2}}\|\nabla\rho\|_{L^{\infty}} +\|\nabla\Delta\rho\|_{L^3}\|\partial^{\gamma}(\frac{1}{1+\rho})\|_{L^6})\\
\geq & \frac{\beta\hbar^2}{36}\|\partial^{\gamma}\nabla\Delta\rho\|^2 -C\beta\hbar^2E\|\nabla\rho\|_{\dot H^s}^2 -C\beta\hbar^2E\|\nabla\rho\|_{\dot H^{s-1}}^2.
\end{split}
\end{equation*}
and by integration by parts and H\"older inequality
\begin{equation*}
\begin{split}
L_6= &\beta\int\frac{1+\theta}{1+\rho}\partial^{\gamma}\Delta\rho \partial^{\gamma}\Delta\rho  +\beta\int[\partial^{\gamma}\nabla\cdot,\frac{1+\theta}{1+\rho}]\nabla\rho \partial^{\gamma}\Delta\rho\\
\geq & \frac{\beta}{18}\|\partial^{\gamma}\Delta\rho\|^2 -C\beta\|\partial^{\gamma}\Delta\rho\| \left(\|\nabla\rho\|_{\dot H^{s-1}}\|\nabla(\rho,\theta)\|_{L^{\infty}} +\|\nabla\rho\|_{L^{\infty}}\|\partial^{\gamma}\nabla(\frac{1+\theta}{1+\rho})\|\right) \\
\geq & \frac{\beta}{36}\|\partial^{\gamma}\Delta\rho\|^2 -C\beta E^2\|\nabla(\rho,\theta)\|^2_{\dot H^{s-1}}.
\end{split}
\end{equation*}
For the term $L_8$, we have
\begin{equation*}
\begin{split}
|L_8|\leq \delta_1\beta\hbar^2\|\partial^{\gamma}\nabla\Delta\rho\|^2 +\frac{C\beta\hbar^2E^2}{\delta_1}\|\nabla\rho\|_{\dot H^{s}}^2.
\end{split}
\end{equation*}
For the term $L_5$, we have by integration by parts
\begin{equation*}
\begin{split}
|L_5|=\left|\beta\int\partial^{\gamma}\Delta\theta\partial^{\gamma}\Delta\rho\right| \leq \beta\delta_1\|\partial^{\gamma}\Delta\rho\|^2 +\frac{\beta}{4\delta_1}\|\partial^{\gamma}\Delta\theta\|^2.
\end{split}
\end{equation*}
For the term $L_4$, we have by integration by parts
\begin{equation*}
\begin{split}
|L_4|=& \left|\beta\int\nabla\partial^{\gamma}(u\cdot\nabla u)\partial^{\gamma}\Delta\rho\right|\\
\leq & \beta\|\partial^{\gamma}\Delta\rho\|(\|u\|_{L^{\infty}}\|\nabla u\|_{\dot H^{s}} +\|\nabla u\|_{\dot H^{s-1}}\|\nabla u\|_{L^{\infty}})\\
\leq & \beta\delta_1\|\partial^{\gamma}\Delta\rho\|^2 +\frac{C\beta E^2}{\delta_1}\|\nabla u\|_{\dot H^{s-1}}^2.
\end{split}
\end{equation*}
For the term $L_3$, we have by integration by parts,
\begin{equation*}
\begin{split}
(\mu+\lambda)^{-1}L_3 =&-\beta\int\partial^{\gamma}\nabla\cdot(\frac{\nabla\text{div}u} {1+\rho})\partial^{\gamma}\Delta\rho\\
=&-\beta\int\frac{\partial^{\gamma}\Delta\text{div}u} {1+\rho}\partial^{\gamma}\Delta\rho -\beta\int[\partial^{\gamma}\nabla,\frac{1}{1+\rho}]\nabla\text{div}u \partial^{\gamma}\Delta\rho\\
= & -\beta\int\frac{\partial^{\gamma}\Delta((1+\rho)\text{div}u)} {(1+\rho)^2}\partial^{\gamma}\Delta\rho +\beta\int\frac{[\partial^{\gamma}\Delta,1+\rho]\text{div}u} {(1+\rho)^2}\partial^{\gamma}\Delta\rho\\
& -\beta\int[\partial^{\gamma}\nabla,\frac{1}{1+\rho}]\nabla\text{div}u \partial^{\gamma}\Delta\rho=:L_{31}+L_{32}+L_{33}.
\end{split}
\end{equation*}
For the first term $L_{31}$, using \eqref{QF1}, it is easy to show by integration by parts,
\begin{equation*}
\begin{split}
L_{31}=& \frac{\beta}{2}\frac{d}{dt}\int\frac{|\partial^{\gamma}\Delta\rho|^2} {(1+\rho)^2} +\beta\int\frac{\rho_t|\partial^{\gamma}\Delta\rho|^2}{(1+\rho)^3} -\frac{\beta}{2}\int\text{div}(\frac{u}{(1+\rho)^2})|\partial^{\gamma}\Delta\rho|^2\\
& +\beta\int\frac{[\partial^{\gamma}\Delta,u]\nabla\rho}{(1+\rho)^2} \partial^{\gamma}\Delta\rho\\
\geq & \frac{\beta}{2}\frac{d}{dt}\int\frac{|\partial^{\gamma}\Delta\rho|^2} {(1+\rho)^2} -C\beta E\|\partial^{\gamma}\Delta(\rho,u)\|^2,
\end{split}
\end{equation*}
thanks to \eqref{equ11} with $p=\infty$ and Lemma \ref{Le-inequ}. By commutator estimates,
\begin{equation*}
\begin{split}
|L_{32}|\leq & C\beta\|\partial^{\gamma}\Delta u\|(\|\nabla\rho\|_{L^{\infty}}\|\text{div}u\|_{\dot H^{s}} +\|\partial^{\gamma}\Delta\rho\|\|\text{div}u\|_{L^{\infty}})\\
\leq & C\beta E\|\nabla(\rho,u)\|_{\dot H^s}^2,
\end{split}
\end{equation*}
and
\begin{equation*}
\begin{split}
|L_{33}|\leq & C\beta\|\partial^{\gamma}\Delta\rho\| (\|\nabla\text{div}u\|_{\dot H^{s-1}}\|\nabla\rho\|_{L^{\infty}} +\|\nabla\text{div}u\|_{L^3}\|\rho\|_{\dot H^{s,6}})\\
\leq & C\beta E\|\nabla(\rho,u)\|_{\dot H^s}^2.
\end{split}
\end{equation*}
The term $L_2$ can be treated similarly and will lead to
\begin{equation*}
\begin{split}
L_2 \geq & \frac{\beta\mu}{2}\frac{d}{dt}\int\frac{|\partial^{\gamma}\Delta\rho|^2} {(1+\rho)^2} -C\beta E\|\nabla(\rho,u)\|_{\dot H^s}^2.
\end{split}
\end{equation*}
Now, we focus on $L_1$. Integrating $L_1$ in time over $[0,t]$, we have
\begin{equation*}
\begin{split}
\beta\int_0^tL_1(s)ds= &\beta\int_0^t\int \partial^{\gamma}\text{div}u_t\partial^{\gamma}\Delta\rho d\tau\\
=& \left.\beta\int_{\Bbb R^3}\partial^{\gamma}\text{div}u\partial^{\gamma}\Delta\rho dx\right|_0^t -\beta\int_0^t\int_{\Bbb R^3}\partial^{\gamma}\text{div}u\partial^{\gamma}\Delta\rho_tdxd\tau\\
=& : I(t)-I(0)-\beta\int_0^tL_{12}(s)ds.
\end{split}
\end{equation*}
By direct estimates, we have
\begin{equation*}
\begin{split}
|I(t)|\leq & \delta_1\beta\|\partial^{\gamma}\Delta\rho(t)\|^2 +\frac{\beta}{\delta_1}\|\partial^{\gamma}\text{div}u(t)\|^2,\\
|I(0)|\leq & \beta\|\partial^{\gamma}\Delta\rho_0\|^2 +\beta\|\partial^{\gamma}\nabla u_0\|^2.
\end{split}
\end{equation*}
For $\beta\int_0^tL_{12}(s)ds$, we have by integration by parts and \eqref{QF1}
\begin{equation*}
\begin{split}
-\beta\int_0^tL_{12}(s)ds
= & -\beta\int_0^t\int_{\Bbb R^3}\nabla\partial^{\gamma}\text{div}u\nabla\partial^{\gamma}(u\cdot\nabla\rho)dxds\\
& -\beta\int_0^t\int_{\Bbb R^3}\nabla\partial^{\gamma}\text{div}u \nabla\partial^{\gamma}[(1+\rho)\text{div}u]dxds\\
\leq & C\beta\int_0^t\|\nabla\partial^{\gamma}\text{div}u\|(\|u\|_{L^{\infty}} \|\nabla\rho\|_{\dot H^{s}} +\|u\|_{\dot H^{s,6}}\|\nabla\rho\|_{L^{3}})d\tau\\
&+C\beta\int_0^t\|\nabla\partial^{\gamma}\text{div}u\| (\|\text{div}u\|_{\dot H^s} +\|\text{div}u\|_{L^{3}}\|\rho\|_{\dot H^{s,6}})d\tau\\
\leq & C\beta\int_0^t\|\nabla\partial^{\gamma}\text{div}u\|^2d\tau +C\beta E^2\int_0^t\|\nabla(\rho,u)\|_{\dot H^s}^2d\tau.
\end{split}
\end{equation*}
Now, we fix some $\delta_1$ small, depending on $\mu$ and $\lambda$, such that $\delta_1=\min\{(2\mu+\lambda)/18,1/142\}$, where $C$ is the constant appearing in the above estimates. Then for such $\delta_1$, fix some small $\varepsilon_0$ such that $\varepsilon_0<\min\{1,1/(72C)\}$. By integrating \eqref{equ41} in time over $[0,t]$, and then taking $E<\varepsilon_0$ small, we complete the proof.
\end{proof}

Now, we prove Proposition \ref{prop3}.
\begin{proof}[Proof of Proposition \ref{prop3}]
Adding the estimates in Lemmas \ref{lem8}-\ref{lem11} together. First, we take $\delta_0={\nu_0}/{4(\mu+\lambda)}$, and then $\hbar_0$ small such that $\hbar_0^2=\min\{\delta_0\nu_0\kappa/4C,\kappa/4C,1\}$, then choose $\beta$ small such that $C\beta\leq C\beta_0:=\min\{\kappa/4,\nu_0/4,1/4\}$, and then for such fixed $\beta$, we choose $\varepsilon_0$ small such that $C\varepsilon_0\leq \min\{\nu_0/4,\kappa/4,\beta/4\}$, then there holds,
\begin{equation*}
\begin{split}
&\|\partial^{\alpha}(\rho,u,\theta,\hbar\nabla\rho)(t)\|^2 +\hbar^2\|\nabla\partial^{\alpha}(u,\hbar\nabla\rho)(t)\|^2 +\beta\|\partial^{\gamma}\Delta\rho(t)\|^2 \\
&+\nu_0\int_0^t\|\partial^{\alpha}\nabla(u,\theta,\hbar\nabla u)(\tau)\|^2d\tau +\beta\int_0^t\|\partial^{\gamma}\Delta(\rho,\hbar\nabla\rho)(\tau)\|^2d\tau \\
\leq & C\|\partial^{\alpha}(\rho,u,\theta,\hbar\nabla\rho)(0)\|^2 +C\hbar^2\|\nabla\partial^{\alpha}(u,\hbar\nabla\rho)(0)\|^2 +C\|\partial^{\gamma}\Delta\rho_0\|^2 \\
&
+C\int_0^t\|\nabla(\rho,u,\theta)\|^2_{\dot H^{s-1}} +C\hbar^2\int_0^t\|\nabla(\rho,u)(\tau)\|_{\dot H^{s}}^2d\tau,
\end{split}
\end{equation*}
for all $E\leq\varepsilon_0$, where $C$ is independent of $t$. Rephrasing this in the ${|||}\cdot|||$ norm, we obtain
\begin{equation*}
\begin{split}
&{|||}\partial^{\alpha}(\rho,u,\theta)(t)|||^2 +\nu_0\int_0^t\|\partial^{\alpha}\nabla(\rho,u,\theta,\hbar\nabla\rho,\hbar\nabla u)(\tau)\|^2d\tau \\
\leq & C{|||}\partial^{\alpha}(\rho,u,\theta)(0)|||^2
+C\int_0^t\|\nabla(\rho,u,\theta)\|^2_{H^{s-1}} +C\hbar^2\int_0^t\|\nabla(\rho,u)(\tau)\|_{\dot H^{s}}^2d\tau.
\end{split}
\end{equation*}
The proof is complete.
\end{proof}

Now, we are ready to obtain \emph{a priori} estimate for the solution of \eqref{QF}.
\begin{theorem}\label{thm3}
Suppose that for some $T>0$, $(\rho, u, \theta)\in\mathcal E_3(0,T)$ is a solution of \eqref{QF} satisfying $E\leq \max_{0\leq t\leq T}{|||}(\rho,u,\theta)(\tau)|||_3\leq \varepsilon$. Then there exists some $\varepsilon_0>0$, $\nu_0=\nu_0(\varepsilon_0)>0$ and $C_0=C_0(\varepsilon_0,\nu_0)$ such that
\begin{equation*}
\begin{split}
&{|||}(\rho,u,\theta)(t)|||_3^2 +\nu_0\int_0^t\|D(\rho,u,\theta,\hbar\nabla\rho,\hbar\nabla u)(\tau)\|_{H^3}^2d\tau \leq C{|||}(\rho,u,\theta)(0)|||_3^2,
\end{split}
\end{equation*}
for any $(\rho,u,\theta)(t)$ satisfying $E\leq \varepsilon_0$.
\end{theorem}

\begin{proof}
Add \eqref{equ43} for $|\alpha|=1$ to $(\frac{C}{\nu_0}+1)$ times \eqref{equ42} to obtain
\begin{equation}\label{equ44}
\begin{split}
{|||}(\rho,u,\theta)(t)|||_1^2 +\nu_0\int_0^t\|D(\rho,u,\theta,\hbar\nabla u,\hbar\nabla\rho)(s)\|_{H^1}^2ds \leq C{|||}(\rho,u,\theta)(0)|||_1^2.
\end{split}
\end{equation}
Then, we add \eqref{equ43} for $|\alpha|=2$ to $(\frac{C}{\nu_0}+1)$ times \eqref{equ44} to obtain
\begin{equation}\label{equ45}
\begin{split}
{|||}(\rho,u,\theta)(t)|||_2^2 +\nu_0\int_0^t\|D(\rho,u,\theta,\hbar\nabla u,\hbar\nabla\rho)(s)\|_{H^2}^2ds \leq C{|||}(\rho,u,\theta)(0)|||_2^2.
\end{split}
\end{equation}
Again, we add \eqref{equ43} for $|\alpha|=3$ to $(\frac{C}{\nu_0}+1)$ times \eqref{equ45}, to obtain
\begin{equation}\label{equ46}
\begin{split}
{|||}(\rho,u,\theta)(t)|||_3^2 +\nu_0\int_0^t\|D(\rho,u,\theta,\hbar\nabla u,\hbar\nabla\rho)(s)\|_{H^3}^2ds \leq C{|||}(\rho,u,\theta)(0)|||_3^2,
\end{split}
\end{equation}
which is the desired estimate, completing the proof.
\end{proof}

Now, we prove Theorem \ref{thm1}.
\begin{proof}[Proof of Theorem \ref{thm1}]
The proof is easy by combining the estimates in Theorem \ref{thm3}, continuity method and the local existence result.
\end{proof}

\section{Proof of Theorem \ref{thm2}}

\subsection{An improved estimate of Matsumura and Nishida \cite{MN80}}
The following estimate was obtained in \cite{MN80}. Under the smallness assumption,
\begin{equation}\label{equ47}
\|(\rho,u,\theta)(t)\|_3^2+\nu_0\int_0^t\|D\rho(\tau)\|_2^2 +\|D(u,\theta)(\tau)\|_3^2d\tau \leq C_0\|(\rho,u,\theta)(0)\|_3^2
\end{equation}
and the following density estimates
\begin{equation}\label{equ48}
\begin{split}
\|D^4\rho(t)\|^2&-C\|D^{3}u(t)\|^2+\nu_0\int_0^t\|D^3\rho(\tau)\|^2d\tau\\
\leq & C\|\rho_0\|_4^2+C\|u_0\|_3^2+C\int_0^t\|D^4(u,\theta)(\tau)\|^2 +\|D(\rho,u,\theta)(\tau)\|^2_{2}d\tau.
\end{split}
\end{equation}
Adding \eqref{equ48} to $(C+\frac{C}{\nu_0}+1)$ times \eqref{equ47}, one obtains
\begin{equation}\label{equ14}
\begin{split}
\|(\rho,u,\theta)(t)\|_3^2& +\|D^4\rho(t)\|^2 +\nu_0\int_0^t\|D(\rho,u,\theta)(\tau)\|_3^2d\tau \leq  C_0\|(\rho,u,\theta)(0)\|_3^2 +C\|\rho_0\|_4^2.
\end{split}
\end{equation}

\subsection{Proof of Theorem \ref{thm2}}
To clearly specify the dependence of the solution on the parameter $\hbar$, we denote $(\rho^{\hbar},u^{\hbar},\theta^{\hbar})$ the solution of the system \eqref{QF} and $(\rho^{0},u^{0},\theta^{0})$ the solution to \eqref{zero}. First of all, from Theorem \ref{thm1} and \eqref{equ14}, we have the estimates
\begin{equation}\label{equ15}
\begin{split}
&\|(\rho^{\hbar},u^{\hbar},\theta^{\hbar})(t)\|_3^2 +\|D^4\rho^{\hbar}(t)\|^2 +\hbar^2\|(\rho^{\hbar},u^{\hbar})(t)\|_4^2 +\hbar^4\|\rho^{\hbar}(t)\|_5^2\\
\leq & C(\|(\rho_0,u_0,\theta_0)(0)\|_3^2 +\|(\rho_0,u_0)(0)\|_4^2 +\|\rho_0\|_5^2).
\end{split}
\end{equation}
and
\begin{equation}
\begin{split}
\|(\rho^0,u^0,\theta^0)(t)\|_3^2& +\|D^4\rho^0(t)\|^2 \leq  C(\|(\rho_0,u_0,\theta_0)\|_3^2 +C\|\rho_0\|_4^2).
\end{split}
\end{equation}

Now, we let
\[N=\rho^{\hbar}-\rho^0,\ U=u^{\hbar}-u^0,\ \Theta=\theta^{\hbar}-\theta^0.
\]
Then $(N,U,\Theta)$ satisfy
\begin{subequations}\label{diff}
\begin{numcases}{}
\partial_tN+\nabla((1+\rho^0)U+Nu^0)=0,\label{d1}\\
\partial_tU-\frac{\mu}{\rho^{\hbar}+1}\Delta U -\left(\frac{\mu}{\rho^{\hbar}+1}-\frac{\mu}{\rho^0+1}\right)\Delta u^0-\frac{\mu+\lambda}{\rho^{\hbar}+1}\nabla\text{div}U -\left(\frac{\mu+\lambda}{\rho^{\hbar}+1}-\frac{\mu+\lambda}{\rho^0+1}\right) \nabla\text{div}u^0\ \ \ \ \ \ \ \ \ \ \nonumber\\
\ \ \ \ \ =-u^{\hbar}\cdot\nabla U -U\cdot\nabla u^0 -\nabla\Theta -\frac{\Theta}{\rho^{\hbar}+1}\nabla\rho^{\hbar} -\frac{\theta^0+1}{\rho^0+1}\nabla N -\left(\frac{\theta^0+1}{\rho^{\hbar}+1}-\frac{\theta^0+1}{\rho^0+1}\right) \nabla\rho^0\nonumber\\
\ \ \ \ \ \ \ \ \ +\frac{\hbar^2}{12}\frac{\Delta\nabla\rho^{\hbar}}{\rho^{\hbar}+1} -\frac{\hbar^2}{3}\frac{\text{div} (\nabla{\sqrt{\rho^{\hbar}+1}} \otimes\nabla{\sqrt{\rho^{\hbar}+1}})}{\rho^{\hbar}+1},\ \label{d2}\\
\partial_t\Theta -\frac{2\kappa}{3(1+\rho^{\hbar})}\Delta\Theta -\frac{2\kappa}{3}\left(\frac{1}{1+\rho^{\hbar}}-\frac{1}{1+\rho^0}\right) \Delta\theta^0 =-u^0\cdot\nabla\Theta -U\cdot\nabla\theta^0\nonumber\\
\ \ \ \ \ \ \ \ \ -\frac23(\theta^{\hbar}+1)\nabla\cdot U -\frac23\Theta\nabla\cdot u^0 +\frac{\hbar^2}{36(1+\rho^{\hbar})}\text{div}((1+\rho^{\hbar}) \Delta u^{\hbar})\nonumber\\
\ \ \ \ \ \ \ \ \ +\frac{2}{3(1+\rho^{\hbar})}\left\{\frac{\mu}{2}(\nabla(u^{\hbar}+u^0) +(\nabla(u^{\hbar}+u^0))^T)(\nabla U+(\nabla U)^T) +\lambda(\text{div}(u^{\hbar}+u^0)\text{div}U\right\}\nonumber\\
\ \ \ \ \ \ \ \ \ +\frac{2}{3}\left(\frac{1}{1+\rho^{\hbar}}-\frac{1}{1+\rho^{0}}\right) \left\{\frac{\mu}{2}|\nabla u^0+(\nabla u^0)^T|^2+\lambda(\text{div} u^0)^2\right\}.\label{d3}\ \ \ \ \ \ \ \
\end{numcases}
\end{subequations}
Now, we multiply \eqref{diff} with $N,U$ and $\Theta$, respectively, integrate the resultant over $\Bbb R^3$ and then sum them up to obtain an energy inequality. Among the many terms, we only treat the following three typical thems in the following. First, for the viscosity term, we have
\begin{equation*}
\begin{split}
-\int\frac{\mu}{\rho^{\hbar}+1}\Delta UU= & \int\frac{\mu}{\rho^{\hbar}+1}|\nabla U|^2 -\int\frac{\mu\nabla\rho^{\hbar}}{(\rho^{\hbar}+1)^2}\nabla UU\\
\geq & \frac{\mu}{2}\int|\nabla U|^2 -C\|U\|^2,
\end{split}
\end{equation*}
where the constant $C$ depends on $\mu$ and the $H^3$ norm of $\rho^{\hbar}$. Secondly, for the last term on the left of \eqref{d2}, we have
\begin{equation*}
\begin{split}
-\int\left(\frac{\mu+\lambda}{\rho^{\hbar}+1}-\frac{\mu+\lambda}{\rho^0+1}\right) \nabla\text{div}u^0U = & \int\left(\frac{(\mu+\lambda)N}{(\rho^{\hbar}+1)(\rho^0+1)}\right)\nabla\text{div}u^0U\\
\leq & C\|\nabla\text{div}u^0\|_{L^3}\|N\|_{L^2}\|U\|_{L^6}\\
\leq & \frac{\mu}{8}\|\nabla U\|^2 +C\|(N,U)\|_{L^2}^2,
\end{split}
\end{equation*}
where the constant $C$ depends on $\mu$ and the $H^3$ norm of $(\rho^{\hbar},\rho^0,u^0)$. Thirdly, for the second to the last term on the RHS of \eqref{d2}, we have
\begin{equation*}
\begin{split}
\frac{\hbar^2}{12}\int\frac{\Delta\nabla\rho^{\hbar}}{\rho^{\hbar}+1}U \leq \frac{\mu}{8}\|U\|^2+\frac{C\hbar^4}{\mu}\|\rho^{\hbar}\|_{H^2}^2\leq \frac{\mu}{8}\|U\|^2+C\hbar^4,
\end{split}
\end{equation*}
where the constant $C$ may depend on $\mu$ and the $H^2$ norm of $\rho^{\hbar}$. The other terms, either depending linearly on the difference $N,U$ or $\Theta$, or depending on the small parameter $\hbar^2$, can be estimated similarly. Therefore, we finally obtain after long but standard estimates
\begin{equation}\label{equ49}
\begin{split}
\frac12\frac{d}{dt}\|(N,U,\Theta)(t)\|_{L^2}^2 +\nu\|\nabla(U,\Theta)\|^2\leq C\|(N,U,\Theta)(t)\|^2+C\hbar^4,
\end{split}
\end{equation}
where $\nu$ depends on the parameters $\mu$ and $\kappa$, and $C$ depends on the $H^3$ norm of $\|(\rho^{\hbar},u^{\hbar},\theta^{\hbar})\|$.

Similarly, taking inner product with $\Delta(N,U,\Theta)$, one can obtain
\begin{equation}\label{equ50}
\begin{split}
\frac12\frac{d}{dt}\|\nabla(N,U,\Theta)(t)\|_{L^2}^2 +\nu\|\Delta(U,\Theta)\|^2\leq C\|\nabla(N,U,\Theta)(t)\|^2+C\hbar^4,
\end{split}
\end{equation}
where $\nu$ also depends on the parameters $\mu$ and $\kappa$, and $C$ depends on the $H^3$ norm of $\|(\rho^{\hbar},u^{\hbar},\theta^{\hbar})\|$. In the derivation of this inequality, we have the following term,
\begin{equation*}
\begin{split}
\int\frac{\hbar^2}{36(1+\rho^{\hbar})}\text{div}((1+\rho^{\hbar})\Delta u^{\hbar})\Delta\Theta\leq &\frac{\kappa}{4}\|\Delta\Theta\|^2 +C\hbar^4\|\text{div}((1+\rho^{\hbar})\Delta u^{\hbar})\|^2\\
\leq &\frac{\kappa}{4}\|\Delta\Theta\|^2 +C\hbar^4,
\end{split}
\end{equation*}
where $C$ depends on $\kappa$ and the $H^3$-norm of $\rho$ and $\Theta$. This procedure can not be proceeded into the higher norms, since the estimates then depends on the $H^4$-norm of $u^{\hbar}$.

Combining the two inequalities \eqref{equ49} and \eqref{equ50}, we obtain
\begin{equation*}
\begin{split}
\frac12\frac{d}{dt}\|(N,U,\Theta)(t)\|_{H^1}^2 +\nu\|\nabla(U,\Theta)\|_{H^1}^2\leq c_1\|(N,U,\Theta)(t)\|_{H^1}^2+c_2\hbar^4,
\end{split}
\end{equation*}
which implies, thanks to the Gronwall inequality, that
\begin{equation*}
\begin{split}
\|(N,U,\Theta)(t)\|_{H^1}^2\leq \left[{c_2}e^{c_1t}/{c_1}\right]\hbar^4.
\end{split}
\end{equation*}
In particular, we note that $c_1$ and $c_2$ are independent of $\hbar$.

We also remark that we can improve the convergence to the $H^2$-norm of $(N,U,\Theta)$ at the price of losing the decay rate. To be precise, we take the inner product of the system \eqref{diff} with $\Delta^2(N,U,\Theta)$ to obtain an energy inequality. Among the terms, we consider the typical term
\begin{equation*}
\begin{split}
&\int\frac{\hbar^2\text{div}((1+\rho^{\hbar})\Delta u^{\hbar})}{36(1+\rho^{\hbar})}\Delta^2\Theta = -\int\nabla\left(\frac{\hbar^2\text{div}((1+\rho^{\hbar})\Delta u^{\hbar})}{36(1+\rho^{\hbar})}\right)\nabla\Delta\Theta\\
\leq & \frac{\kappa}{4}\|\nabla\Delta\Theta\|^2+ \hbar^4\Big(\|\nabla\text{div}\Delta u^{\hbar}\|_{L^2}^2 +\|\Delta\rho^{\hbar}\|_{L^{6}}^2\|\Delta u^{\hbar}\|_{L^3}^2\\
& +\|\nabla\rho^{\hbar}\|_{L^{\infty}}^2\|\text{div}\Delta u^{\hbar}\|_{L^2}^2 +\|\nabla\rho^{\hbar}\|_{L^{\infty}}^4\|\Delta u^{\hbar}\|_{L^2}^2 \Big)\\
\leq & \frac{\kappa}{4}\|\nabla\Delta\Theta\|^2 +\hbar^2\Big(\|\hbar\nabla\text{div}\Delta u^{\hbar}\|_{L^2}^2 +(1+\|\rho^{\hbar}\|_{H^3}^4)\|u^{\hbar}\|_{H^3}^2\Big)\\
\leq & \frac{\kappa}{4}\|\nabla\Delta\Theta\|^2 +C\hbar^2,
\end{split}
\end{equation*}
thanks to the estimate in \eqref{equ15}. The other terms, either depending linearly on $(N,U,\Theta)$ or depending on the Planck constant $\hbar$, can be treated similarly by making use of \eqref{equ15}. Finally, we arrive at the inequality
\begin{equation*}
\begin{split}
\frac12\frac{d}{dt}\|(N,U,\Theta)(t)\|_{H^2}^2 +\nu\|\nabla(U,\Theta)\|_{H^2}^2\leq c_1\|(N,U,\Theta)(t)\|_{H^2}^2+c_2\hbar^2,
\end{split}
\end{equation*}
which implies, thanks to the Gronwall inequality, that
\begin{equation*}
\begin{split}
\|(N,U,\Theta)(t)\|_{H^2}^2\leq \left[{c_2}e^{c_1t}/{c_1}\right]\hbar^2.
\end{split}
\end{equation*}
In particular, $c_1$ and $c_2$ are independent of $\hbar$. This completes the proof.

\end{CJK*}

\begin{thebibliography}{00}

\bibitem{AI89} M.G. Ancona and G.J. Iafrate, Quantum correction to the equation of state of an electron gas in semiconductor, Phys. Rev. B, 39, (1989)9536-9540.

\bibitem{AT87} M.G. Ancona and H.F. Tiersten, Macroscopic physics of the silicon inversion layer, Phys. Rev. B, 35, (1987)7959-7965.

\bibitem{BYZ14} D. Bian, L. Yao and C. Zhu, Vanishing capillarity limit of the compressible fluid models of Korteweg type to the Navier-Stokes equations, SIAM J. Math. Anal., 46(2), (2014)1633-1650.

\bibitem{Bohm52} D. Bohm, A suggested interpretation of the quantum theory in terms of ``hidden" valuables: I; II, Phys. Rev., 85, (1952)166-179; 180-193.

\bibitem{DS85} J.E. Dunn and J. Serrin, On the thermodynamics of interstitial working, Arch. Ration. Mech. Anal., 88, (1985)95-133.

\bibitem{Feynman72} R. Feynman, Statistical Mechanics, a Set of Lectures, New York: W.A. Benjamin, 1972.

\bibitem{Gardner94} C.L. Gardner, The quantum hydrodynamic model for semiconductor devices, SIAM J. Appl. Math., 54(2), (1994)409-427.

\bibitem{Haas11} F. Haas, Quantum plasmas: An hydrodynamic approach, Springer, New York, 2011.

\bibitem{HLi94} H. Hattori, D. Li, Solutions for two-dimensional system for materials of Korteweg type, SIAM J. Math. Anal., 25(2), (1994)85-98.

\bibitem{HLi96} H. Hattori, D. Li, Global solutions of a high dimensional system for Korteweg materials, J. Math. Anal. Appl., 198, (1996)84-97.

\bibitem{Jungel10} A. Jungel, Global weak solutions to compressible Navier-Stokes equations for quantum fluids, SIAM J. Math. Anal., 42(3), (2010)1025-1045.

\bibitem{JLW14} A. Jungel, C.-K. Lin and K.-C. Wu, An asymptotic limit of a Navier-Stokes system with capillary effects, Comm. Math. Phys., 329, (2014)725-744.

\bibitem{HL09} L. Hsiao and H. Li, The well-posedness and asymptotics of multi-dimensional quantum hydrodynamics, Acta Math. Sci., 29B(3), (2009)552-568.

\bibitem{KP88} T. Kato and G. Ponce, Commutator estimates and the Euler and Navier-Stokes equations, Comm. Pure Appl. Math., 41, (1988)891-907.

\bibitem{Korteweg1901} D. Korteweg, Sur la forme que prennent les \'{e}quations du mouvement des fluides si l'on tient compte des forces capillaires par des variations de densit\'{e}. Arch. N\'{e}er. Sci. Exactes S\'{e}r,  II 6, (1901)1-24.

\bibitem{LL05} H. Li and C.K. Lin, Zero Debye length asymptotic of the quantum hydrodynamic model for semiconductors, Comm. Math. Phys., 256(1), (2005)195-212.

\bibitem{LM04} H. Li and P. Marcati, Existence and asymptotic behavior of multi-dimensional quantum hydrodynamic model for semiconductors, Comm. Math. Phys., 245(20), (2004)215-247.

\bibitem{LM01} H. Li and P. Markowich, A review of hydrodynamical models for semiconductors: asymptotic behavior, Bol. Soc. Brasil Mat., 32(3), (2001)321-342.

\bibitem{MN80} A. Matsumura, T. Nishida, The initial value problem for the equations of motion of viscous and heat-conductive gases, J. Math. Kyoto Univ., 20-1, (1980)67-104.

\bibitem{Nirenberg59} L. Nirenberg, On elliptic partial differential equations, Ann. Scuola Norm. Sup. Pisa, 13(3), (1959)115-162.

\bibitem{Pu13} X. Pu, Dispersive limit of the Euler-Poisson system in higher dimensions, SIAM J. Math. Anal., 45(2), (2013)834-878.

\bibitem{PG13} X. Pu and B. Guo, Global existence and convergence rates of smooth solutions for the full compressible MHD equations, Z. Angew. Math. Phys., 64, (2013)519-538.

\bibitem{WT11} Y. Wang and Z. Tan, Optimal decay rates for the compressible fluid model of Korteweg type, J. Math. Anal. Appl., 379, (2011)256-271.

\bibitem{Wigner32} E. Wigner, On the quantum correction for thermodynamic equilibrium, Phys. Rev., 40, (1932)749-759.


\end{thebibliography}
\end{document}